\documentclass[letterpaper]{article}

\usepackage[a4paper,includefoot,margin=3cm]{geometry}

\usepackage[T1]{fontenc}
\usepackage[utf8]{inputenc}
\usepackage{amsmath,amssymb,amsthm}
\usepackage{mathtools}
\usepackage{microtype}
\usepackage{graphicx}
\usepackage{xcolor}
\usepackage[hidelinks]{hyperref}
\usepackage{orcidlink}

\usepackage{tikz}
\usetikzlibrary{patterns,decorations.pathreplacing}

\tikzset{
    bvertex/.style={circle,draw,fill=black,inner sep=1.5pt},
    bsvertex/.style={circle,draw, inner sep=1.5pt},
    avertex/.style={draw,fill=lightgray,inner sep=2pt}
}

\usepackage{complexity}
\usepackage{enumitem}

\DeclareMathOperator*{\argmin}{argmin}

\DeclareMathOperator{\seg}{seg}
\DeclareMathOperator{\opt}{opt}

\newcommand{\NN}{\mathbb{N}}
\newlang{\TP}{3PART}

\DeclareMathOperator{\Cater}{\mathbf{Cater}}

\usepackage{mathtools}
\DeclarePairedDelimiter{\ceil}{\lceil}{\rceil}

\DeclareMathOperator{\cost}{cost}
\DeclareMathOperator{\dist}{dist}

\usepackage[disable]{todonotes}

\newcommand{\gic}{\textsc{Graph in Caterpillar}}
\newcommand{\gicShort}{\textsc{GiC}}

\newcommand{\tic}{\textsc{Tree in Caterpillar}}
\newcommand{\ticShort}{\textsc{TiC}}

\usepackage{etoolbox}

\newcommand{\appsymb}{$\star$}
\newcommand{\appref}[1]{{\hyperref[proof:#1]{\appsymb}}}

\newcommand{\OptproblemDef}[3]{%
\begin{center}
	\setlength{\tabcolsep}{2pt}
	\begin{tabular}{@{}lp{10cm}@{}}
		\multicolumn{2}{@{}l}{\textsc{#1}} \\% Title
		\textbf{Input:} & #2 \\% Input
		\textbf{Task:} & #3 \\% Question
	\end{tabular}
\end{center}%
}

\usepackage[capitalise]{cleveref}

\usepackage{thmtools}
\theoremstyle{plain}
\newtheorem{theorem}{Theorem}
\newtheorem{lemma}[theorem]{Lemma}

\newtheorem{corollary}[theorem]{Corollary}
\newtheorem{claim}[theorem]{Claim}

\theoremstyle{definition}

\theoremstyle{remark}

\newtheorem{remark}[theorem]{Remark}

\usepackage{authblk}

\title{Designing Caterpillars for Graphs: Approximation and Hardness}
\author[1]{Leon Kullmann}
\author[1]{Phuoc Lucky Trinh}
\author[2]{Leon Kellerhals}
\author[3]{Mitja Krebs}
\author[1]{André Nichterlein}
\author[3,4]{Stefan Schmid}

\affil[1]{Algorithmics and Computational Complexity, Technische Universität Berlin, Germany}
\affil[2]{Computational Intelligence Research Group, Technische Universität Clausthal, Germany}
\affil[3]{Intelligent Networks, Technische Universität Berlin, Germany}
\affil[4]{Fraunhofer SIT, Germany}

\begin{document}

\maketitle

% TODO mandatory: add short abstract of the document
\begin{abstract}
	The classical \textsc{Minimum Linear Arrangement} (MLA) problem has been studied extensively.
    It is known to be \NP-hard and it admits an $O(\sqrt{\log n}\log\log n)$-approximation [Feige and Lee, IPL, 2007].
	MLA can be defined as follows as design problem:
	Given a graph \(G\) with vertex set~$V(G)$, design a path~$H$ on the same vertex set that minimizes the linear arrangement cost~\(\sum_{uv\in E(G)}\dist_H(u,v)\), where~$\dist_H(u,v)$ indicates the distance of~$u$ and~$v$ in~$H$.
% 	A linear arrangement of a graph \(G\) is a path \(H\) on vertex set \(V(G)\), and its cost is \(\sum_{uv\in E(G)}\dist_H(u,v)\).
% 	Minimizing this cost---
% 	but almost always with \(H\) fixed to be a path.
	We initiate the study of the generalization in which \(H\) is allowed to be a caterpillar graph of maximum degree at most~\(\Delta\).
	Caterpillars are the simplest generalization of paths, having pathwidth one and interpolating between paths and stars via the degree parameter \(\Delta\).
	We give an algorithm that lifts any \(\alpha\)-approximation for MLA to an \((\alpha+3-2/(\Delta-1))\)-approximation for our problem, thus obtaining an $O(\sqrt{\log n}\log\log n)$-approximation for our more general problem as well.
% 	, and a ratio approaching \(\alpha\) as the optimum grows.
%	, thereby also achieving a similar approximation factor in general.
	Moreover, we derive a \(4\)-approximation whenever MLA is polynomial-time solvable, in particular, for trees.
% 	On the hardness side, 
	Complementing these results, we prove \NP-hardness for every constant \(\Delta\geq 2\), and, in stark contrast to MLA, show it remains \NP-hard on trees when \(\Delta\) is part of the input.
\end{abstract}

% \todo[inline]{
% 	Page limit: 12 pages excluding references, title page, and appendix.\\
%   General TODOs:\\
%   - Edge notation \(\{u,v\}\) versus \(uv\): DONE\\
%   - Notation for distance \(\dist\) versus \(d\): DONE \\
%   - cost vs. costs: DONE \\
%   - Problem names: in \textsc{TextSC}! : DONE\\
%   - Please not too long lines (git tries to merge line by line).\\
%   - make sure to use \texttt{\textbackslash[...\textbackslash]} instead of double dollars; especially when using lipics.
% }

% \newpage

\section{Introduction}

A \emph{linear arrangement} (also called \emph{linear layout}) of a graph \(G\) on \(n\) vertices is a bijection~\(\pi\colon V(G)\to[n]\).
% , that is, an arrangement of the vertices of \(G\) on a line.
Linear arrangement problems have been studied for over half a century, motivated by applications ranging from VLSI circuit design and numerical linear algebra to scheduling, graph drawing, and computational biology~\cite{diaz2002survey}.
There are various objectives that quantify how well the edges of $G$ are preserved by such an arrangement.
Three objectives have received particular attention:
the \emph{bandwidth} (or dilation), \(\max_{uv\in E(G)}|\pi(u)-\pi(v)|\);
the \emph{cutwidth} (or congestion), the maximum number of edges of \(G\) crossing any cut of the line; and
the \emph{linear arrangement cost}, \(\sum_{uv\in E(G)}|\pi(u)-\pi(v)|\).
The computational problems minimizing the corresponding objective are called \textsc{Bandwidth}, \textsc{Cutwidth}, and \textsc{Minimum Linear Arrangement} (MLA).

All three problems are \NP-hard on arbitrary input graphs~\cite{garey1974some,gavril1977some,papadimitriou1976npcompleteness}, and a substantial body of work has mapped out their complexity as one varies the input graph \(G\).
When \(G\)~is a tree, for instance, the picture splits:
\textsc{Bandwidth} remains \NP-hard~\cite{garey1979computers}, whereas both \textsc{Cutwidth}~\cite{yannakakis1985polynomial} and MLA~\cite{chung1984optimal} admit polynomial-time algorithms.
Tractable or well-approximable cases of MLA are known for hypercubes, rectangular grids, dense graphs, and planar graphs, and a long line of work on general graphs has driven the best approximation ratio down to \(O(\sqrt{\log n}\log\log n)\) due to Feige and Lee~\cite{feige2007improved}.

% Equivalently, 
% A linear arrangement of \(G\) can be seen as a bijection of the vertices of~$G$ into the vertices of a path~$H$ with~$n$ vertices.
An equivalent definition can be given from the viewpoint of network design problems:
Find a path \(H\) with the same vertex set as~$G$ that minimizes the respective objective.
The three objectives can then be recast as \(\max_{uv\in E(G)}\dist_{H}(u,v)\), the maximum edge-cut over cuts of \(H\), and \(\sum_{uv\in E(G)}\dist_{H}(u,v)\), respectively.
From this vantage point it is natural to ask what happens if $H$ can be chosen from a richer class of host graphs?
This direction has recently gained traction (see related work).
In this paper, we focus on the linear arrangement objective and allow \(H\) to be a \emph{caterpillar}, i.e., a tree whose non-leaf vertices induce a path, called the \emph{spine} of the caterpillar, with the remaining leaves (the \emph{legs}) attached to spine vertices.
Caterpillars are a natural starting point for at least two reasons.
Structurally, they are precisely the connected graphs of pathwidth one~\cite{proskurowski1999classes}. %, and deleting all legs returns the spine.
In this sense they form the ``closest'' generalization of paths.
Algorithmically, they give any network design procedure an additional degree of freedom over paths:
each vertex can now be placed either on the spine or as a leg of some already-placed spine vertex.
% Finally, caterpillars come with a natural parameter, the maximum degree \(\Delta\) of the host, which interpolates between paths (\(\Delta=2\)) and stars \(\Delta=n-1\)

% \todo{lke: picture?}
Note that irrespective of~$H$ the cost~\(\sum_{uv\in E(G)}\dist_{H}(u,v)\) is always at least~$|E(G)|$.
Since a star (a caterpillar with one spine vertex) yields a cost of at most~$2|E(G)|$ for any~$G$, our problem is trivially 2-approximable.
Thus, having great flexibility in the choice of $H$ can lead to better approximation guarantee; this contrasts the non-trivial approximation algorithms for MLA with much higher approximation factor~\cite{feige2007improved}.
% \mkcom{Should we move this to the introduction, perhaps as a third reason why caterpillars are a natural starting point?}
% 
For a better control of the caterpillar, we extend the input by a bound~$\Delta>1$ on the maximum degree of the desired caterpillar (i.e., we look for a degree-$\Delta$ caterpillar). 
The parameter~$\Delta$ allows for a fine control on how ``MLA-like''~$H$ has to be. 
Formally, we consider the following problem: 

\OptproblemDef{\gic{} (\gicShort{})} 
{An undirected graph~$G$ and an integer~$\Delta > 1$.}
{Find a degree-$\Delta$ caterpillar~\(H\) on vertex set~$V(G)$ minimizing~$\sum_{uv\in E(G)}\dist_{H}(u,v)$.}

\paragraph{Related Work.}
Hardness of MLA on general graphs was established by Garey, Johnson, and Stockmeyer~\cite{garey1974some} via a reduction from \textsc{Max Cut}.
For MLA on trees, Goldberg and Klipker~\cite{goldberg1976algorithm} gave the first polynomial-time algorithm, with subsequent improvements by Shiloach~\cite{shiloach1979minimum} and Chung~\cite{chung1984optimal}.
Currently, MLA on an \(n\)-vertex tree can be solved in \(O(n^{\log 3/\log 2})\) time.
On the approximation side, Leighton and Rao~\cite{leighton1999multicommodity} obtained an \(O(\log^2 n)\)-approximation for MLA on general graphs, improved to \(O(\log n \log\log n)\) by Even et al.~\cite{even2000divideandconquer} via spreading metrics, then to \(O(\log n)\) by Rao and Richa~\cite{rao1998new}, and finally to \(O(\sqrt{\log n}\log\log n)\) by Feige and Lee~\cite{feige2007improved}.
Planar graphs admit an \(O(\log\log n)\)-approximation~\cite{rao1998new}, and dense graphs admit a PTAS~\cite{arora1996new}.
Exact optimal solutions are also known for several structured graph classes, including hypercubes~\cite{harper1964optimal} and square grids~\cite{mitchison1986optimal}.
See Díaz, Petit, and Serna~\cite{diaz2002survey} for further pointers.
All of these results restrict \(H\) to be a path.

The only prior work we are aware of on varying \(H\) falls into two strands.
A first line of work considers designing cycles~\cite{raspaud2000congestion} or grids~\cite{bezrukov2000congestion} under dilation/congestion measures.
Closer to our setting, a second line studies the linear arrangement cost for other host classes.
Avin, Mondal, and Schmid~\cite{avin2020demand} gave an \(O(1)\)-approximation when \(H\) is required to be a degree-\(4\) graph and \(G\) is a tree;
Kellerhals et al.~\cite{kellerhals2026designingapproximatebinarytrees} gave a \(4\)-approximation for the case where \(H\) is a binary tree and \(G\) is a tree;
Schmid et al.~\cite{schmid2015splaynet} gave a polynomial-time algorithm for the case where \(H\) is a binary search tree and \(G\) is arbitrary; and 
Liebchen and Wünsch~\cite{LW08} gave an overview on different problem variants where the task is to find a spanning tree~$H$ of~$G$ under various objectives, including linear arrangement cost.
% Johnson, Lenstra, and Rinnooy Kan~\cite{johnson1978complexity} proved \NP-hardness when \(H\) is required to be a spanning tree of \(G\), see

\paragraph{Our Contribution.}
We introduce and study \gicShort{} as a generalization of MLA.
\gicShort{} inherits its \NP-hardness from its special case MLA (that is \gicShort{} with~$\Delta = 2$).
In \cref{ssec:gicHardness}, we extend this \NP-hardness to all values~$\Delta \ge 2$.
Somewhat surprisingly, we also extend any approximation result of MLA to \gicShort{} in \cref{sec:approx}. 
More precisely, our algorithm takes any \(\alpha\)-approximate linear arrangement and rounds it into a degree-\(\Delta\) caterpillar whose cost is within an additive \(O(|E(G)|)\) term of \(\alpha\) times the optimum.
This implies that for any graph class for which MLA is polynomial-time solvable (e.g.\ if~$G$ is a tree), \gicShort{} admits a polynomial-time $4$-approximation algorithm.

We also discover some differences in the computational complexity of MLA and \gicShort{}. 
Notably, MLA is polynomial-time solvable if~$G$ is a tree while we prove in \cref{sec:tic} that \gicShort{} remains \NP-hard, even if~$G$ is a tree.
This indicates that trying to transfer polynomial-time (exact) algorithms from MLA to \gicShort{} might be a pointless endeavor. 

\paragraph{Preliminaries.}
All graphs are assumed to be undirected.
For a given graph $G$, we denote its vertices by $V(G)$ and its edges by $E(G)$. 
%Let $\Graph$ be the class of all finite graphs.
%We write $H \subseteq G$ if $H$ is a \emph{subgraph} of $G$, that is, $V(H) \subseteq V(G)$ and $E(H) \subseteq E(G)$, and we write $H \leq G$ if $H$ is an \emph{induced subgraph} of $G$, that is, $V(H) \subseteq V(G)$ and $E(H) = E(G) \cap (V(H) \times V(H))$.  For graphs $G,H \in \Graph$, we define their \emph{union} as $$G \cup H := (V(G) \cup V(H), E(G)\cup E(H)).$$ We can w.l.o.g.\ assume that $G$ and $H$ have disjoint vertex sets. Then, we call their union a \emph{disjoint union} and write $G \sqcup H := G \cup H$.
For $u,v \in V(G)$, we denote by $\dist_G(u,v)$ the length of a shortest $uv$-path if one exists and infinity otherwise.
The \emph{cost} of a graph $H$ with $V(H) = V(G)$ is
\[ \cost(G,H) := \sum_{uv \in E(G)} \dist_H(u,v). \]
We say an edge $uv \in E(G)$ is \emph{stretched} by $H$ if the stretch amount $\dist_H(u,v)-1$ is greater than $0$.
%For graph classes $\mathcal{G},\mathcal{H} \subseteq \Graph$, we study the problem of embedding graphs $G \in \mathcal{G}$ into graphs $H \in \mathcal{H}$ defined as follows.
%\problem{
%    Graph Embedding of $\mathcal{G}$ into $\mathcal{H}$ ($\GE[\mathcal{G}, \mathcal{H}$])
%}{
%    A graph $G \in \mathcal{G}$ and a number $W \in \mathbb{N}$.
%}{
%    Is there a graph $H \in \mathcal{H}$ with $V(H) = V(G)$ that satisfies $\cost(G,H) := \sum_{\{u,v\} \in E(G)} d_H(u,v) \leq W$?
%}
For a graph class $\mathcal{H}$, we define $\opt_{\mathcal{H}}(G) \coloneqq \min\{ \cost(G, H) \mid H \in \mathcal{H},\, V(G) = V(H) \}$,
and we say that $H \in \mathcal{H}$ is \emph{optimal} if $V(G) = V(H)$ and the cost of $H$ is $\opt_{\mathcal{H}}(G)$.
%More generally, if $ V(G) \neq V(H)$ but $ |V(G)| = |V(H)|$, we can still talk about embeddings of $G$ into $H$ as follows:
%We say $G$ \emph{embeds into} $H$ with costs at most $W$ if there is a bijection $\pi : V(G) \to V(H)$ such that
%\[\cost_\pi(G,H):= \sum_{uv \in E(G)} \dist_H(\pi(u),\pi(v)) \leq W.\]
%For graph classes $\mathcal{G},\mathcal{H} \subseteq \Graph$, we will also consider a variation of graph embedding where we give the graph we want to embed into as input.
%\problem{
%    Graph Embedding with Input Graph of $\mathcal{H}$ into $\mathcal{G}$ ($\GE_\mathrm{input}[\mathcal{G}, \mathcal{H}$])
%}{
%    Graphs $G \in \mathcal{G}, H \in \mathcal{H}$ with $|V(G)|=|V(H)|$ and a number $W \in \mathbb{N}$.
%}{
%    Is there a bijection $f: V(G) \to V(H)$ such that 
%    $\cost_f(G,H):= \sum_{\{u,v\} \in E(G)} d_H(f(u),f(v)) \leq W$?
%}
%Again, we can talk about optimal embeddings but now with respect to bijections from  $V(G)$ to $V(H)$. A bijection $\pi: V(G) \to V(H)$ is an \emph{optimal embedding} of $G$ into $H$ if 
%\[\pi \in \argmin \{\cost_\sigma(G,H) ~|~ \sigma : V(G) \to V(H) \text{ is bijective}\}.\]
%Moreover, the \emph{optimal cost} of embedding $G$ into $H$ is given by
%\[\opt(G,H) := \min \{\cost_\pi(G,H)~|~ \pi : V(G) \to V(H) \text{ is bijective}\}.\]

A \emph{caterpillar} is a tree $C$
whose non-leaf vertices induce a path, called the \emph{spine}.
A leaf $u$ in $C$ is called \emph{leg} of spine vertex $v$, where $v$ is the unique neighbor of $u$.
For a spine vertex $v$, the \emph{segment} $\seg_C(v)$ of $v$ is the subgraph of $C$ induced by $v$ and all legs of $v$.
We write $\seg(v)$ if $C$ is clear from the context.
A caterpillar is a \emph{degree-$\Delta$ caterpillar} if it has maximum degree at most $\Delta$. 
A caterpillar is \emph{$\Delta$-legged} if each of its segments has at most $\Delta$ legs.
Note that the only difference between restricting the degree of a caterpillar and restricting the number of legs a segment can have is whether the outermost segments are allowed to have one more leg than the inner segments or not.
We call a segment of a degree-$\Delta$ or $\Delta$-legged caterpillar \emph{full} if its spine vertex has the maximum number of legs.
A degree-$\Delta$ or $\Delta$-legged caterpillar is \emph{full} if each but one outermost segment is full.
%Let $\Cater$ be the class of all caterpillars, and let $\Cater_\Delta$ be the class of all degree-$\Delta$ caterpillars.

\section{Designing Caterpillars for Graphs}\label{sec:gic}

In this section, we first show our main result: an~$O(\sqrt{\log n}\log\log n)$-approximation for \gicShort{} which is based on the known approximation for MLA~\cite{feige2007improved}.
Our approximation works for any given degree bound~$\Delta$ for the caterpillar to be found. 
We subsequently show \NP-hardness of \gicShort{} for any given constant~$\Delta$.

\subsection{Approximation Algorithm}\label{sec:approx}

\newcommand{\flatten}{\textnormal{\textsc{Flatten}}}
\newcommand{\starfy}{\textnormal{\textsc{Starfy}}}

We next show that any \(\alpha\)-approximation for the MLA problem can be lifted to an (\(\alpha+3-2/(\Delta-1)\))-approximation for \gicShort{}.
% the problem of optimally embedding a graph into a degree-\(\Delta\) caterpillar.
Combined with the \(O(\sqrt{\log n}\log\log n)\)-approximation for MLA due to Feige and Lee~\cite{feige2007improved}, this yields an approximation of the same order for our problem.

For \(\Delta=n-1\), a star already serves as a trivial \(2\)-approximation, since every edge of \(G\) is stretched to a path of length at most \(2\) and the optimum is at least \(|E(G)|\).
Our algorithm extends this idea to arbitrary \(\Delta\ge 2\):
it builds stars of size \(\Delta-1\) and concatenates them along a spine.
The stars are dictated by a linear arrangement \(\pi\) of \(G\):
the \(i\)-th star contains those vertices placed by \(\pi\) into positions \((i-1)(\Delta-1)+1,\dots,i(\Delta-1)\).
For each star, its center is chosen to be a vertex of maximum degree in \(G\).

We make this formal as follows.
A linear arrangement of \(G\) is a bijection \(\pi\colon V(G)\to[n]\).
For a linear arrangement \(\pi\) and a vertex \(v\), let
\[
  \beta_\pi(v):=\ceil*{\frac{\pi(v)}{\Delta-1}}
\]
be the index of the star that \(v\) belongs to.

\paragraph{\starfy{} algorithm.}
Given a linear arrangement \(\pi\), proceed as follows.
\begin{enumerate}
  \item Partition \(V(G)\) into \(\ceil*{n/(\Delta-1)}\) consecutive segments \(S_i:=\{v\in V(G)\mid\beta_\pi(v)=i\}\).
  \item For each segment \(S_i\), choose a vertex \(s_i\in S_i\) of maximum degree in \(G\) and turn \(S_i\) into a star centered at \(s_i\).
  \item Connect the centers along a spine by adding the edges \(s_{i-1}s_i\) for all \(i>1\).
\end{enumerate}
% The resulting caterpillar \(H\) has maximum degree \(\Delta\): the centers \(s_i\) have at most \(\Delta-2\) leaves and at most two spine neighbors.

We analyze \starfy{} in two steps.
First, we bound the cost of the caterpillar \(H\) produced by \starfy{} in terms of the cost of the input arrangement \(\pi\) (\Cref{lem:starfy}).
Second, we bound the cost of an MLA from above by the cost of a path obtained by \emph{flattening} an optimal caterpillar (\Cref{lem:flatten}).
Combining the two bounds yields the main result (\Cref{thm:starfy}).

\begin{lemma}\label{lem:starfy}
  Let \(\pi\) be a linear arrangement of \(G\), and let \(H\) be the caterpillar produced by \starfy{} on input \(\pi\). Then
  \[
    \sum_{uv\in E(G)}\dist_{H}(u,v)
    \leq
    \sum_{uv\in E(G)}\left|\beta_\pi(u)-\beta_\pi(v)\right|
    +\left(2-\frac{2}{\Delta-1}\right)|E(G)|.
  \]
\end{lemma}

\begin{proof}
  Fix an edge \(uv\in E(G)\).
  In \(H\), the unique \(uv\)-path decomposes into three pieces: from \(u\) to the center \(s_{\beta_\pi(u)}\) of its segment, along the spine from \(s_{\beta_\pi(u)}\) to \(s_{\beta_\pi(v)}\), and from \(s_{\beta_\pi(v)}\) to \(v\).
  The first and third pieces have length at most \(1\), and length \(0\) precisely when \(u\) (respectively \(v\)) is the center of its segment.
  The middle piece has length \(|\beta_\pi(u)-\beta_\pi(v)|\).
  Hence
  \[
    \dist_{H}(u,v)
    \leq\left|\beta_\pi(u)-\beta_\pi(v)\right|
    +[u\neq s_{\beta_\pi(u)}]
    +[v\neq s_{\beta_\pi(v)}],
  \]
  where \([\cdot]\) denotes the Iverson bracket:
  For a proposition \(P\), \([P]=1\) if \(P\) is true and \([P]=0\) otherwise.

  Summing over all edges \(uv\in E(G)\), the contribution of the two indicator terms is
  \[
    \sum_{uv\in E(G)}[u\neq s_{\beta_\pi(u)}]+[v\neq s_{\beta_\pi(v)}]
    =\sum_{v\in V(G)}\deg_G(v)\cdot[v\neq s_{\beta_\pi(v)}]
    =2|E(G)|-\sum_{i}\deg_G(s_i).
  \]
  Since each center \(s_i\) is chosen to have maximum degree in its segment \(S_i\),
  \[
    \sum_{i}\deg_G(s_i)
    \geq\sum_{i}\frac{1}{|S_i|}\sum_{v\in S_i}\deg_G(v)
    \geq\frac{1}{\Delta-1}\sum_{v\in V(G)}\deg_G(v)
    =\frac{2|E(G)|}{\Delta-1}.
  \]
  Combining these two displays yields
  \[
    \sum_{uv\in E(G)}\dist_H(u,v)
    \leq\sum_{uv\in E(G)}\left|\beta_\pi(u)-\beta_\pi(v)\right|
    +2|E(G)|-\frac{2|E(G)|}{\Delta-1}.\qedhere
  \]
\end{proof}

To bound the cost of an MLA, we bound the cost of a specific linear arrangement obtained by linearizing an optimal caterpillar.

\paragraph{\flatten{} algorithm.}
Given a caterpillar \(H\) with spine \(s_1 \dots s_k\) (in some fixed left-to-right order), construct a path \(H'\) on \(V(H)\) as follows:
for each spine vertex \(s_i\) with \(i<k\), insert the legs of \(s_i\) immediately to the left of \(s_i\) (in arbitrary order);
insert the legs of \(s_k\) immediately to the right of \(s_k\).
The result is a path \(H'\) on all \(|V(H)|\) vertices in which the spine vertices appear in their original left-to-right order, each preceded by its legs (except \(s_k\), which is followed by its legs).

\begin{lemma}\label{lem:flatten}
  Let \(H\) be a degree-\(\Delta\) caterpillar, and let \(H'\) be the path obtained from \(H\) via \flatten{}.
  Then for every pair of vertices \(u,v\in V(H)\),
  \[
    \dist_{H'}(u,v)\leq(\Delta-1)\cdot\dist_{H}(u,v).
  \]
\end{lemma}

\begin{proof}
  Let \(P=v_0v_1\cdots v_\ell\) be the unique \(uv\)-path in \(H\), with \(v_0=u\) and \(v_\ell=v\), oriented from left to right with respect to the orientation fixed by \flatten{}.
  (If \(u\) and \(v\) land in the same spine segment, orient arbitrarily.)
  After \flatten{}, the vertices \(v_0,v_1,\dots,v_\ell\) appear in this same left-to-right order along \(H'\), possibly with additional legs of \(v_1,\dots,v_\ell\) interspersed between consecutive vertices of \(P\).
  For each \(i\in[\ell]\), if \(v_i\) is an interior spine vertex, it has at most \(\Delta-2\) legs (since it has two spine neighbors); if \(v_i\) is the rightmost spine vertex, \flatten{} places its legs to its right, so none are inserted before \(v_i\); and if \(v_i\) is a leg of \(H\), it has no legs of its own.
  In every case, at most \(\Delta-2\) vertices are inserted between \(v_{i-1}\) and \(v_i\).
  (The legs of \(v_0 = u\) are placed to the left of \(v_0\) by \flatten{} and therefore do not appear between \(v_0\) and \(v_\ell\).)

  Summing over all \(\ell\) steps,
  \[
    \dist_{H'}(u,v)
    \leq \ell + (\Delta-2)\,\ell
    = (\Delta-1)\cdot\dist_H(u,v). \qedhere
  \]
\end{proof}

We now combine the two lemmas.

\begin{theorem}\label{thm:starfy}
  Let \(H^{*}\) be an optimal degree-\(\Delta\) caterpillar for \(G\), let \(\pi\) be an \(\alpha\)-approximation for MLA on \(G\), and let \(H\) be the caterpillar produced by \starfy{} on input~\(\pi\). Then
  \[
    \sum_{uv\in E(G)}\dist_{H}(u,v)
    \leq\alpha\sum_{uv\in E(G)}\dist_{H^{*}}(u,v)+\left(3-\frac{2}{\Delta-1}\right)|E(G)|.
  \]
\end{theorem}

\begin{proof}
  By \Cref{lem:starfy} and the bound \(|\beta_\pi(u)-\beta_\pi(v)|\leq |\pi(u)-\pi(v)|/(\Delta-1)+1\),
  \begin{align*}
    \sum_{uv\in E(G)}\dist_{H}(u,v)
    &\leq\sum_{uv\in E(G)}\left(\frac{|\pi(u)-\pi(v)|}{\Delta-1}+1\right)+\left(2-\frac{2}{\Delta-1}\right)|E(G)|\\
    &=\frac{1}{\Delta-1}\sum_{uv\in E(G)}|\pi(u)-\pi(v)|
    +\left(3-\frac{2}{\Delta-1}\right)|E(G)|.
  \end{align*}
  Let \(\pi^{*}\) be an MLA of \(G\). Since \(\pi\) is an \(\alpha\)-approximation,
  \[
    \sum_{uv\in E(G)}\dist_{H}(u,v)
    \leq\frac{\alpha}{\Delta-1}\sum_{uv\in E(G)}|\pi^{*}(u)-\pi^{*}(v)|+\left(3-\frac{2}{\Delta-1}\right)|E(G)|.
  \]
  Let \((H^{*})'\) be the path obtained from \(H^*\) via \flatten{}, and let \(\pi'\) be its induced linear arrangement. Since \(\pi^{*}\) is optimal,
  \[
    \sum_{uv\in E(G)}|\pi^{*}(u)-\pi^{*}(v)|
    \leq\sum_{uv\in E(G)}|\pi'(u)-\pi'(v)|
    =\sum_{uv\in E(G)}\dist_{(H^{*})'}(u,v).
  \]
  By \Cref{lem:flatten}, \(\dist_{(H^{*})'}(u,v)\leq(\Delta-1)\dist_{H^{*}}(u,v)\) for every \(uv\in E(G)\), hence
  \begin{align*}
    \sum_{uv\in E(G)}\dist_{H}(u,v)
    &\leq\frac{\alpha}{\Delta-1}\sum_{uv\in E(G)}(\Delta-1)\dist_{H^{*}}(u,v)+\left(3-\frac{2}{\Delta-1}\right)|E(G)|\\
    &=\alpha\sum_{uv\in E(G)}\dist_{H^{*}}(u,v)
    +\left(3-\frac{2}{\Delta-1}\right)|E(G)|.\qedhere
  \end{align*}
\end{proof}

Using the trivial lower bound \(\sum_{uv\in E(G)}\dist_{H^{*}}(u,v)\geq|E(G)|\) in \Cref{thm:starfy} yields
\[
  \sum_{uv\in E(G)}\dist_{H}(u,v)
  \leq\left(\alpha+3-\frac{2}{\Delta-1}\right)\sum_{uv\in E(G)}\dist_{H^{*}}(u,v),
\]
which is what we invoke in \Cref{corollary:starfy,corollary:factor-4-approx}.
The bound of \Cref{thm:starfy} however degrades gracefully with the optimum:
when \(\sum_{uv\in E(G)}\dist_{H^{*}}(u,v)=\omega(|E(G)|)\), the additive term is \(o(1)\), and \starfy{} is an \((\alpha+o(1))\)-approximation.

The additive slack also vanishes at the extreme value \(\Delta=2\):
here \starfy{} returns the input linear arrangement unchanged.
Since \(\Delta-1=1\) divides \(n\), the ceiling terms in the proof of \Cref{lem:starfy} contribute no rounding, so the approximation factor reduces to \(\alpha\).

% The additive \(\Theta(|E(G)|)\) slack is unavoidable for our analysis, however, as the other extreme \(\Delta=n-1\) witnesses.
% Here \starfy{} produces a single star:
% it creates two segments \(S_1,S_2\), but \(S_2\) consists of a single vertex whose only neighbor is the spine vertex \(s_1\) of \(S_1\), and so it is absorbed as a leg of \(s_1\).
% Consequently \(\beta_\pi(u)=\beta_\pi(v)\) for every edge \(uv \in E(G)\), and \Cref{lem:starfy} yields an approximation factor of \(2-2/(n-2)\), independent of \(\alpha\).

Combining \Cref{thm:starfy} with the best known approximation for MLA~\cite{feige2007improved} yields the following.

\begin{corollary}\label{corollary:starfy}
  For every \(\Delta\ge 2\), there is an \(O(\sqrt{\log n}\log\log n)\)-approximation algorithm for \gic{}.
  % the problem of optimally embedding a graph into a degree-\(\Delta\) caterpillar.
\end{corollary}

\begin{proof}
  Run the \(O(\sqrt{\log n}\log\log n)\)-approximation algorithm for MLA of Feige and Lee~\cite{feige2007improved} to obtain a linear arrangement \(\pi\), and apply \starfy{} to \(\pi\).
  The claim follows from \Cref{thm:starfy} together with the lower bound \(\sum_{uv\in E(G)}\dist_{H^{*}}(u,v)\geq|E(G)|\).
\end{proof}

Plugging in known (approximation) algorithms for MLA on restricted graph classes, we obtain the following specific ratios.

\begin{corollary}\label{corollary:factor-4-approx}
  \starfy{} yields
  a polynomial-time $\alpha$-approximation algorithm for \gicShort{}
  with 
  \begin{itemize}
    \item \(\alpha = 4-2/(\Delta-1)\), when \(G\) is a tree, a hypercube, or a square grid; and
    \item \(\alpha = O(\log\log n)\), when \(G\) is planar.
  \end{itemize}
\end{corollary}

\begin{proof}
  MLA is polynomial-time solvable for trees~\cite{chung1984optimal}, hypercubes~\cite{harper1964optimal}, and square grids~\cite{mitchison1986optimal}
  and admits a polynomial-time~\(O(\log \log n)\)-approximation for planar graphs~\cite{rao1998new}.
  Invoking \Cref{thm:starfy} with the lower bound \(\sum_{uv\in E(G)}\dist_{H^{*}}(u,v)\geq|E(G)|\) yields the claimed approximation algorithm for \gicShort{}.
\end{proof}

% \section{Hardness Results}\label{sec:hardness}

\subsection{NP-hardness} \label{ssec:gicHardness}
% Does a 1- or 2-vertex caterpillar not have any central vertex?
In this section, we show that the problem of designing optimally degree-$\Delta$ caterpillars for arbitrary graphs is $\NP$-hard for every constant $\Delta \ge 2$. 
For $\Delta < 2$, this problem becomes trivial. 
We generalize the $\NP$-hardness proof of Garey et al.~\cite{garey1974some} for \textsc{Minimum Linear Arrangement} via a reduction from \textsc{Max Cut}. 
% They show that for $\Delta =2$ the problem of optimally embedding graphs into paths is $\NP$-hard 
In \textsc{Max Cut} we are given a graph $G=(V,E)$ and a number~$k \in \NN$ and the task is to find a partition $S_1 \cup S_2 = V$ such that there are at least $k$ edges from $S_1$ to $S_2$.% We extend their results to all $\Delta \geq 2$ by generalizing their reduction.

We first sketch the \NP-hardness proof of \textsc{Minimum Linear Arrangement} (MLA) from Garey et al.~\cite{garey1974some}.
Recall that, given a graph $G' = (V', E')$ and $W \in \mathbb N$, MLA asks whether there is an embedding $\pi \colon V(G') \to V(H)$ of cost at most~$W$, where $H$ is a path on $|V(G)|$ vertices.
Given an instance $(G=(V,E),k)$ of \textsc{Max Cut}, we construct the graph $G'=(V',E')$ with $W \in \mathbb{N}$ as follows. 
The graph $G'$ is the complement graph of $G$ together with~$r=n^4$ many vertices $\{u_1,\dots,u_r\} = U$ that are connected to all other vertices of $G'$. Set $W=\opt(K_{r+n}) - kr$ where $\opt(K_{r+n})$ denotes the cost of an optimal embedding of $K_{r+n}$ into a path. 
The idea of the reduction is the following. To minimize the cost of an embedding we want to find non-edges in $E'$, that is edges in $E$, and place the endpoints as far as possible from each other on the path. This will correspond to a partition $S_1 \cup S_2 = V$. 
The additional $n^4$ vertices are required to outweigh any cost contribution from edges within a partition so that an optimal embedding is as shown in \Cref{fig:minla_path}.
More formally, if $S_1,S_2$ is a partition of $V$ such that $|\{uv \in E \mid u \in S_1, v \in S_2\}| \ge k$, then we can embed $G'$ into a path $P$ with cost at most $W$ by placing $S_1$ (in arbitrary order) leftmost in $P$, $S_2$ (in arbitrary order) rightmost in $P$ and $U$ in between of $S_1$ and $S_2$. Conversely, if there is an embedding of $G'$ into a path $P$ with cost at most $W$ then we can show that there is another embedding of $G'$ into a path $P'$ such that no $v \in V$ has vertices $u,u' \in U$ both to the left and right of it by swapping certain vertices. This allows us to form a partition $S_1,S_2$. Then, by elementary computations and the choice of $r$ we can show that $|\{uv \in E \mid u \in S_1, v \in S_2\}| \ge k$.

We generalize the reduction as follows. Instead of swapping vertices we swap whole segments. Additionally, we need to impose an upper bound on the maximum distance of two vertices in the caterpillar. Without this constraint, the bound on the size of the cut is not large enough.

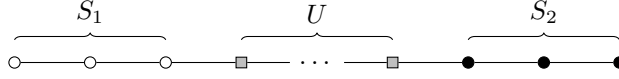
\begin{figure}[t]
    \centering
    \begin{tikzpicture}[
        clique/.style={draw,fill=lightgray,inner sep=2pt},
        vertexs1/.style={circle,draw,inner sep=1.5pt},
        vertexs2/.style={circle,draw,fill=black,inner sep=1.5pt}
    ]
    
    % central nodes
    \node[vertexs1] (v0) at (-1,0) {};
    \node[vertexs1] (v1) at (0,0) {};
    \node[vertexs1] (v2) at (1,0) {};
    \node[clique] (u1) at (2,0) {};
    \node (u2) at (3,0) {$\dots$};
    \node[clique] (u3) at (4,0) {};
    \node[vertexs2] (v3) at (5,0) {};
    \node[vertexs2] (v4) at (6,0) {};
    \node[vertexs2] (v5) at (7,0) {};

    % edges
    \draw (v0)--(v1)--(v2)--(u1)--(u2)--(u3)--(v3)--(v4)--(v5);

    \draw [
            decoration={brace,raise=0.3cm},decorate
        ] (-1,0) -- (1,0) 
        node [pos=0.5,anchor=south,yshift=0.4cm] {$S_1$};
    \draw [
            decoration={brace,raise=0.3cm},decorate
        ] (2,0) -- (4,0) 
        node [pos=0.5,anchor=south,yshift=0.4cm] {$U$};
    \draw [
            decoration={brace,raise=0.3cm},decorate
        ] (5,0) -- (7,0) 
        node [pos=0.5,anchor=south,yshift=0.4cm] {$S_2$};

    \end{tikzpicture}
    \caption{
    A linear arrangement witnessing that $G'=(V',E')$ and $W$ is a YES-instance for $\textsc{MLA}$. Square vertices denote vertices from $U$ and circular vertices denote vertices from $V$. White vertices belong to $S_1$ and black vertices to $S_2$. The distance between vertices from $S_1$ to $S_2$ is at least $|U| = r$.}
    \label{fig:minla_path}
\end{figure}

%\problem{Max Cut ($\MC$)}
%{A graph $G=(V,E)$ and a number $k \in \mathbb{N}$.}
%{Is there a partition of the vertices $S_1 \cup S_2 = V$ with cut at least $k$, that is,  $|\{\{u,v\} \in E \mid u \in S_1, v \in S_2\}| \ge k$?}

%\begin{theorem}[\cite{karp2009reducibility}]
    %$\MC$ is \NP-complete.
%\end{theorem}

%\problem{
    %Graph Embedding into Degree-$\Delta$ Caterpillars 
    %($\GEC_\Delta$)
%}{
    %A graph $G=(V,E)$ and a number $W \in \mathbb{N}$.
%}{
    %Is there a degree-$\Delta$ caterpillar $C=(V,E')$ such that $$\cost(G,C) = \sum_{\{u,v\} \in E}d_C(u,v) \le W?$$
    %\vspace{-8pt}
%}

\begin{theorem}% [\appref{thm:gec-NP-hardness}]\footnote{Proofs of statements marked with \appsymb{} are deferred to the appendix.} 
	\gic{} remains $\NP$-hard for each fixed~$\Delta \geq 2$.
	\label{thm:gec-NP-hardness}
%     For all $\Delta \geq 2$, $\GEC_\Delta$ is $\NP$-hard.
\end{theorem}
% \appendixproof{thm:gec-NP-hardness}
% {
\begin{proof}
    We reduce from \textsc{Max Cut}.
    
    \textbf{Reduction.}
    On input $G=(V,E)$ and $k \in \NN$, the reduction does the following. Let 
    \begin{align*}
    \ell &:= \binom{n}{2} \cdot (4n+7) + 1, \\
    %n'&:=-n ~\;\mathrm{mod}\, (\Delta-1) ~ \in \{0,\dots,\Delta-2\}, \\
    r&:=\ell (\Delta-1)+2.%+n'.
    \end{align*}
    We will come back to the semantics of these numbers later.
    Let $U = \{u_1,\dots,u_r\}$  such that $U \cap V = \varnothing$. The output of the reduction is $(G'=(V',E'),W)$ with
    \begin{align*}
        V'&:=V \cup U, \\
        E'&:=\{uv \mid u,v \in V' \text{, } u \neq v \text{ and } uv \notin E\}, \\
        W&:=\opt_{\Cater_\Delta}(K_{n+r})-k\ell,
    \end{align*}
    where $\Cater_\Delta$ denotes the class of all degree-$\Delta$ caterpillars. %Note that by the choice of $r$ and $n'$ $\Delta-1$ divides $|V'|-2$.
    We call $u \in U$ a \emph{clique vertex}.
    
    \textbf{Correctness.} Let $G=(V,E)$ and $k$ be a YES-instance for \textsc{Max Cut}. Then there is a partition $S_1 \cup S_2 = V$ with $|\{uv \in E \mid u \in S_1, v \in S_2\}| \ge k$. Arrange the vertices $V'$ in a full caterpillar $C$ as follows. The vertices of $S_1$ are placed leftmost in $C$, the vertices of $S_2$ rightmost and the vertices of $U$ in between. Whenever a segment contains vertices from $S_1$ or $S_2$ and $U$, we choose a clique vertex as the spine vertex of that segment. %As $|V'| - 2$ is divisible by $\Delta -1$, all segments of $C$ are full so that the choice of putting $S_1$ to the left and $S_2$ to the right of the vertices of $U$ was arbitrary. 
    \Cref{fig:fullcat} shows the construction.

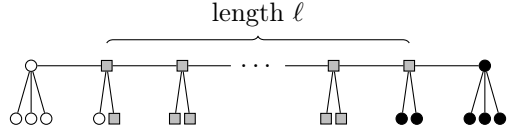
\begin{figure}
        \centering
        \begin{tikzpicture}[
            clique/.style={draw,fill=lightgray,inner sep=2pt},
            vertexs1/.style={circle,draw,inner sep=1.5pt},
            vertexs2/.style={circle,draw,fill=black,inner sep=1.5pt}
        ]
        
        % central nodes
        \node[vertexs1] (v1) at (0,0) {};
        \node[clique] (v2) at (1,0) {};
        \node[clique] (u1) at (2,0) {};
        \node (u2) at (3,0) {$\dots$};
        \node[clique] (u3) at (4,0) {};
        \node[clique] (v3) at (5,0) {};
        \node[vertexs2] (v4) at (6,0) {};
        
        % leaves
        \node[vertexs1] (v1a) at (-0.2,-0.7) {};
        \node[vertexs1] (v1b) at (0,-0.7) {};
        \node[vertexs1] (v1c) at (0.2,-0.7) {};
        \node[vertexs1] (v2a) at (0.9,-0.7) {};
        \node[clique] (v2b) at (1.1,-0.7) {};
        \node[clique] (u1a) at (1.9,-0.7) {};
        \node[clique] (u1b) at (2.1,-0.7) {};
       % \node[clique] (u2a) at (2.9,-0.7) {};
        %\node[clique] (u2b) at (3.1,-0.7) {};
        \node[clique] (u3a) at (3.9,-0.7) {};
        \node[clique] (u3b) at (4.1,-0.7) {};
        \node[vertexs2] (v3a) at (4.9,-0.7) {};
        \node[vertexs2] (v3b) at (5.1,-0.7) {};
        \node[vertexs2] (v4a) at (5.8,-0.7) {};
        \node[vertexs2] (v4b) at (6,-0.7) {};
        \node[vertexs2] (v4c) at (6.2,-0.7) {};

        % edges
        \draw (v1)--(v2)--(u1)--(u2)--(u3)--(v3)--(v4);
        
        \draw (v1)--(v1a);
        \draw (v1)--(v1b);
        \draw (v1)--(v1c);

        \draw (v2)--(v2a);
        \draw (v2)--(v2b);

        \draw (u1)--(u1a);
        \draw (u1)--(u1b);

        %\draw (u2)--(u2a);
        %\draw (u2)--(u2b);

        \draw (u3)--(u3a);
        \draw (u3)--(u3b);
        
        \draw (v3)--(v3a);
        \draw (v3)--(v3b);
        
        \draw (v4)--(v4a);
        \draw (v4)--(v4b);
        \draw (v4)--(v4c);

        \draw [
                decoration={brace,raise=0.3cm},decorate
            ] (1,0) -- (5,0) 
            node [pos=0.5,anchor=south,yshift=0.4cm] {length $\ell$};
        
        \end{tikzpicture}
        \caption{A full degree-$\Delta$ caterpillar witnessing that $G'=(V',E')$ and $W$ is a YES-instance for $\gicShort$. Square vertices denote clique vertices and circular vertices denote vertices from $V$. White vertices belong to $S_1$ and black vertices to $S_2$.}
        \label{fig:fullcat}
    \end{figure}

    Then,
    \begin{align*}
        \cost(G',C) = \sum_{uv \in E'}\dist_C(u,v) &= \opt_{\Cater_\Delta}(K_{n+r}) - \sum_{uv \notin E'}\dist_C(u,v) \\
        &= \opt_{\Cater_\Delta}(K_{n+r}) - \sum_{uv \in E}\dist_C(u,v) \\ &\le \opt_{\Cater_\Delta}(K_{n+r})-k\ell = W.
    \end{align*}
    The second equality holds as $C$ is a full caterpillar. The last inequality holds as $|\{uv \in E \mid u \in S_1, v \in S_2\}| \ge k$ and $\dist_C(u,v) \ge \ell$ for $u \in S_1, v \in S_2$.

    Let $G'=(V',E')$ and $W$ be a YES-instance for $\gicShort$ as obtained by the reduction. Then there is a caterpillar $C_0$ such that $\cost(G',C_0) = \sum_{uv \in E'}\dist_{C_0}(u,v) \le W$.

    First, we show a useful lemma that tells us how to optimally embed a clique into a degree-$\Delta$ caterpillar.
\begin{lemma}
    Let $N \in \NN$. Every optimal degree-$\Delta$ caterpillar $C$ for $K_N$ is full.
    \label{prop:opt_clique}
\end{lemma}
\begin{proof}
    For $N \in \{1,2\}$, the optimal caterpillar clearly is $C = K_N$.
    Let $N > 2$.
    Assume $C = (V,E)$ is a non-full degree-$\Delta$ caterpillar that is optimal for $K_N$. Let $s_0 s_1 \dots s_\ell$ be the spine of $C$. For all $0 \leq j \leq \ell$, let $r_j$ be the number of legs of $s_j$. Assume without loss of generality that $r_0 \leq r_\ell$. We pick a leg $v$ from $\seg_C(s_0)$ and give it to the closest non-full segment $\seg(s_i), i\neq 0$, that is, we construct a new degree-$\Delta$ caterpillar 
    \[C' = (V, (E \cup \{vs_i\}) \setminus \{ vs_0\}).\] 
    We claim that $\cost(K_N, C') < \cost(K_N, C)$ which is a contradiction to the optimality of $C$. To see this, partition 
    \[V \setminus \{v\} = V(\seg_{C'}(s_0)) \cup L \cup V(\seg_C(s_i))\cup R\] 
    where \[L = \bigcup_{j=1}^{i-1} V(\seg_C(s_j)) \quad \text{and} \quad R = \bigcup_{k=i+1}^\ell V(\seg_C(s_k)).\] Note that the distances between any two vertices in $ V \setminus \{v\}$ in $C$ and $C'$ are the same. As $L$ induces a full caterpillar in $C$ and in $C'$, the sum over distances from $v$ to nodes in $L$ is also the same in $C$ and $C'$ by symmetry. However, the distances between $v$ and the vertices in $\seg_{C'}(s_0)$ increases in $C'$, and the distances between $v$ and vertices in $\seg_C(s_i)\cup R$ decreases in $C'$. Concretely, we have
    \begin{align*}
       \cost(K_N, C')= &\cost(K_N, C) 
       + \underbrace{(r_0 - 1) \cdot i + i}_{\text{dist.\ to } \seg_{C'}(s_0)}
       - \underbrace{(r_i \cdot i + i)}_{\text{dist.\ to } \seg_{C}(s_i)} 
       - \underbrace{i \cdot |R|}_{\text{dist.\ to } R} \\
       = &\cost(K_N, C) + (r_0 - r_i -1 -|R|) \cdot i .
    \end{align*}
    If $i = \ell$, then $$(r_0 - r_i -1 -|R|) \cdot i = (r_0 - r_\ell - 1) \cdot \ell <0$$ by assumption. If $i < \ell$, then $R$ contains $\seg(s_\ell)$ and again
    \[(r_0 - r_i -1 -|R|) \cdot i \leq (r_0 - r_i -1 -(r_\ell + 1)) \cdot i < 0\] by assumption. Thus, $\cost(K_N, C') < \cost(K_N, C)$.
\end{proof}
    
    For a caterpillar $C$, let $s^C_1 \dots s_{t_C}^C$ be its spine. Let $C(i) := \seg(s^C_i)$ be the $i$th segment of $C$ for all $i \leq t_C$. Define 
    $$\alpha_C : V' \to [t_C], \quad \alpha_C(v) = i \quad \text{if } v \in V(C(i)).$$ A segment $C(i)$ is called \emph{clique segment} if $V(C(i)) \subseteq U$. We show that we may assume that $C_0$ has at most $\ell + n + 3$ clique segments.

    \begin{claim}
        We may assume that $C_0$ has at most $\ell + n + 3$ clique segments.
    \end{claim}
    \begin{proof}
        We transform $C_0$ into a caterpillar with at most $\ell + n + 3$ clique segments. A \emph{clique segment block} $$B=C_0 \left[\bigcup_{i \le j \le i+x} V(C_0(j)) \right]$$ is a subgraph of $C_0$ induced by consecutive clique segments such that $C_0(i-1)$ and $C_0(i+x+1)$ (if they exist) are not clique segments. The proof is split into three parts.
        \begin{enumerate}
            \item Clique segment blocks are full caterpillars.
            \item There are at most $n+1$ clique segment blocks.
            \item If $C_0$ had more than $\ell + n + 3$ clique segments, then $C_0$ contains strictly more than $r$ clique vertices which is a contradiction.
        \end{enumerate}
        \textit{Part 1.} Let $$B=C_0 \left [ \bigcup_{i \le j \le i+x} V(C_0(j)) \right]$$ be a clique segment block. Let $$L=\bigcup_{1 \le j \le i-1}V(C_0(j)) \quad \text{ and } \quad R=\bigcup_{i+x+1 \le j \le t_{C_0}}V(C_0(j)).$$ Assume $|L| \le |R|$. If $B$ is not a full caterpillar yet, then transform $B$ into a full caterpillar as follows: (a) Iteratively select a leg of the leftmost segment of $B$ that contains at least one leg. Make this leg a leg of the rightmost segment of $B$ that is not full. (b) If (a) is not applicable anymore, then select the rightmost spine $s_j$ of $B$ without any legs and make it a leg of segment $C_0(j+1)$ if $C_0(j+1)$ is not full and connect $s_{j-1}$ (if $s_{j-1}$ exists) with $s_{j+1}$. If $C_0(j+1)$ is full, then make $s_{j-1}$ (if $s_{j-1} \in V(B)$) a leg of $s_j$ and connect $s_{j-2}$ (if $s_{j-2}$ exists) with $s_j$. It is easy to see that exhaustively applying (a) and then exhaustively applying (b) turns $B$ into a full caterpillar.
        
        We analyze the new cost incurred by the edges $e = uv \in E'$. First, consider edges $e \in E' \setminus V(B)^2$ not within $B$. Each application of (a) to a leg $u \in V(B)$ reduces the cost by $d$ for each $v \in R$ and increases the cost by $d$ for each $v \in L$ where $d$ denotes the distance between the old spine vertex and the new spine vertex of $u$. The cost of edges $u'v'$ with $u' \in L, v' \in R$ do not change. By assumption $|L| \le |R|$. Thus, $d \cdot |L| - d \cdot |R| \le 0$ and the cost contributed by edges $e \in E' \setminus V(B)^2$ decrease when we apply (a). Each application of (b) to a spine vertex $u \in V(B)$ increases the cost by $1$ for each $v \in L$ and decreases the cost by $0$ for each $v \in R$. The cost of edges $u'v' \in E'$ with $u' \in L, v' \in R$ decrease by $1$ as the length of the spine of $B$ decreases by $1$. As $1 \cdot |L| - 0 \cdot |R| - 1 \cdot |L| \cdot |R| \le 0$, the cost contributed by edges $e \in E' \setminus V(B)^2$ decrease. Next, consider edges $e \in E' \cap V(B)^2$ within $B$. After all applications of (a) and (b) $B$ is a full caterpillar. Thus, the contributed cost decrease by \Cref{prop:opt_clique}.
        
        \textit{Part 2.} Clique segment blocks can only be separated by segments containing a vertex $v \in V$. As $|V|=n$, there are at most $n+1$ clique segment blocks.
        
        \textit{Part 3.} Assume for contradiction that $C_0$ has at least $\ell + n + 4$ clique segments. By Part~1 we may assume that all clique segment blocks are full caterpillars. By Part~2 there are at most $n+1$ clique segment blocks, each of which contains at most one non-full clique segment. Thus, $C_0$ contains at least $\ell + n + 4 - (n+1) = \ell + 3$ full clique segments. Each full clique segment contains at least $\Delta - 1$ clique vertices. Then, $C_0$ has at least
        $$
            (\ell + 3)(\Delta - 1) = \ell (\Delta - 1) + 3\Delta - 3 > \ell (\Delta - 1) + 2 = r
        $$
        clique vertices which is a contradiction.
    \end{proof}
    
    Let $\mathcal{C}$ denote the set of all degree-$\Delta$ caterpillars on the vertex set $V'$ that have at most $\ell + n + 3$ clique segments. Observe that for any $C \in \mathcal{C}$
    \begin{align*}
        \sum_{uv \in E'}\dist_C(u,v) + \sum_{uv \in E}\dist_C(u,v) &= \sum_{u,v \in V'}\dist_C(u,v) \\
        &\ge \opt_{\Cater_\Delta} (K_{n+r}) = W + k\ell.
    \end{align*}
    For $C_0 \in \mathcal{C}$ we have $$\sum_{uv \in E}\dist_{C_0}(u,v) \ge W - \sum_{uv \in E'}\dist_{C_0}(u,v) + k\ell \ge k\ell$$ as $W - \sum_{uv \in E'}\dist_{C_0}(u,v) = W - \cost(G', C_0) \ge 0$ by assumption.

    Define 
    \begin{align*}
    W^* &= \max_{C \in \mathcal{C}} \sum_{uv \in E}\dist_C(u,v), \text{ and} \\
    \mathcal{C}^* &= \{C \in \mathcal{C} \mid \sum_{uv \in E}\dist_C(u,v) = W^*\}.    
    \end{align*}
    Clearly, $W^* \ge k\ell$ and $\mathcal{C}^* \neq \varnothing$ (as witnessed by $C_0$).

    We now show that there is a $C \in \mathcal{C}^*$ such that no vertex $v \in V$ has clique segments $C(i)$ and $C(i')$ to the left and right, respectively, that is, 
    $i < \alpha_C(v) < i'$. 
    This will allow us to form a partition $S_1\cup S_2 = V$.
    
    For each $C \in \mathcal{C}^*$ define 
    \begin{align*}        
    S(C) = \{v \in V \mid 
    &\exists i,i' \in [t_C] : C(i),C(i') \text{ are clique segments and } \\
    &i < \alpha_C(v) < i'\}
    \end{align*}
    as the set of vertices of $V$ not having the above property and define $m(C) = |S(C)|$. Let $C^* \in \mathcal{C}^*$ which minimizes $m(C^*)$, that is, $m(C^*) \le m(C)$, for all $C \in \mathcal{C}^*$. We claim that $m(C^*)=0$.

\begin{claim}
    Let $C^* \in \argmin_{C \in \mathcal{C}^*} m(C)$. Then $m(C^*)=0$.
\end{claim}
\begin{proof}
    Assume for contradiction that $m(C^*)>0$. Then, $S(C^*) \neq \varnothing$. Let $v_0$ be the rightmost vertex of $S(C^*)$, that is, $\alpha_{C^*}(v_0) \ge \alpha_{C^*}(v)$, for all $v \in S(C^*)$ and let $i_0 = \alpha_{C^*}(v_0)$.

    For each $v \in V'$ define
    \begin{align*}
        L(v) &= |\{w \in V \mid vw \in E , \alpha_{C^*}(w) < \alpha_{C^*}(v)\}| \\
        R(v) &= |\{w \in V \mid vw \in E , \alpha_{C^*}(w) > \alpha_{C^*}(v)\}|
    \end{align*}
    and define $L(i)=\sum_{\alpha_{C^*}(v)=i}L(v)$, $R(i)=\sum_{\alpha_{C^*}(v)=i}R(v)$, for $i \in [t_{C^*}]$.

    Note that $u \in U$ implies $L(u)=R(u)=0$, and $C^*(i)$ is a clique segment implies $L(i)=R(i)=0$.

    \textit{Case 1.} Assume that $L(i_0) \ge R(i_0)$. Let $C^*(j_0)$ be the rightmost clique segment, that is, $j_0 \ge i$, for all $i \in [t_{C^*}]$ such that $C^*(i)$ is a clique segment. By definition of $v_0$, every $C^*(i)$ with $i_0 < i \le j_0$ is a clique segment. Consider the caterpillar $C \in \mathcal{C}$ obtained from $C^*$ by swapping the segments $C^*(i_0)$ and $C^*(j_0)$. Then, as $L(i_0) \ge R(i_0)$ we have that $\sum_{uv \in E}\dist_C(u,v) \ge W^*$ and thus $C \in \mathcal{C}^*$ but $m(C) < m(C^*)$ which is a contradiction to the choice of $C^*$.

    \textit{Case 2.} Assume that $L(i_0) < R(i_0)$. Let $$t = \max\{i \in [t_{C^*}] \mid i < i_0 \text{ and } L(i) \ge R(i)\}.$$ The value of $t$ is well-defined since there exists a clique segment $C^*(i)$, with $i < i_0$. By definition of $t$ and as $L(i_0) < R(i_0)$ we have that $L(t+1) < R(t+1)$. Consider the caterpillar $C \in \mathcal{C}$ obtained from $C^*$ by swapping the segments $C^*(t)$ and $C^*(t+1)$. Then $\sum_{uv \in E}\dist_C(u,v) > W^*$ which is a contradiction to the definition of $W^*$.
    
    Thus, $m(C^*)=0$.
\end{proof}

    Define the sets
    \begin{align*}
        S_1&= \{v \in V \mid \nexists i \in [t_{C^*}] : C^*(i) \text{ is a clique segment and } i < \alpha_{C^*}(v)\} \\
        S_2&= \{v \in V \mid \nexists i \in [t_{C^*}] : C^*(i) \text{ is a clique segment and } i > \alpha_{C^*}(v)\}.
    \end{align*}
    We show that $S_1,S_2 \subseteq V$ is a partition of $V$.

\begin{claim}
    The sets $S_1,S_2 \subseteq V$ as defined above from a partition of $V$.
\end{claim}
\begin{proof}
    Assume for contradiction that $S_1 \cup S_2 \neq V$. Let $v \in V$ but $v \notin S_1 \cup S_2$. Then $v \in S(C^*)$ and $m(C^*) > 0$, a contradiction.
\iffalse
    \begin{claim}
        There are at least $\ell$ many clique segments.
    \end{claim}
    \begin{proof}
        There are at most $n$ many non-clique segments. Each of these non-clique segments contains at most $(\Delta - 1)$ many clique vertices. Hence, at least $r - (\Delta - 1) n$ many clique vertices are in clique segments. Each inner segment contains at most $\Delta-1$ many vertices and the outer segments contain at most $\Delta$ many vertices. To get a lower bound on the number of clique segments, we divide $r - (\Delta -1)n-2$ many clique vertices into $\Delta -1$ sized segments---two of which can be outer segments and contain an additional leg each. Thus, these $r - (\Delta -1)n$ clique vertices form at least 
        \begin{align*}
            \frac{r - (\Delta -1)n-2}{\Delta - 1} 
            &= \frac{(\ell + n) (\Delta-1)+2+n' - (\Delta -1)n-2}{\Delta - 1} \\ 
            &= \frac{\ell(\Delta -1) + (\Delta-1)n - (\Delta -1)n + 2-2 + n'}{\Delta - 1} \\
            &= \frac{\ell(\Delta -1) + n'}{\Delta - 1} \ge \ell 
        \end{align*}
        many clique segments. 
    \end{proof}
\fi

    Now, we show $S_1 \cap S_2 = \varnothing$. Each segment that is not a clique segment contains at most $\Delta - 1$ clique vertices except for the two outermost segments that might contain $\Delta$ clique vertices. As $\ell > n$ and $r \ge \ell (\Delta - 1) + 2$ there is at least one clique segment. Thus, $S_1 \cap S_2 = \varnothing$.
\end{proof}

    As $C^* \in \mathcal{C}^*$ it holds that $k\ell \le \sum_{uv \in E}\dist_{C^*}(u,v)$ and $C^*$ has at most $\ell + n + 3$ clique segments. For $i \in \{1,2\}$, the distance between two vertices in $S_i$ is at most $n+1$ because all vertices of $S_i$ appear in consecutive non-clique segments of which there are at most $n$. The distance between any two vertices in $C^*$ is at most $\ell + 2n + 5$. We have
    \begin{align*}
        k\ell &\le \sum_{uv \in E}\dist_{C^*}(u,v) \\
        &= \sum_{\substack{uv \in E \\ u,v \in S_1}}\dist_{C^*}(u,v) + \sum_{\substack{uv \in E \\ u,v \in S_2}}\dist_{C^*}(u,v) + \sum_{\substack{uv \in E \\ u \in S_1,v \in S_2}}\dist_{C^*}(u,v) \\
        &\le \binom{n}{2} \cdot (n+1) + \binom{n}{2} \cdot (n+1) + (\ell + 2n + 5) \cdot |\{uv \in E \mid u \in S_1,v \in S_2\}| \\
        &\le \binom{n}{2} \cdot (2n+2) + \binom{n}{2} \cdot (2n+5) + \ell \cdot |\{uv \in E \mid u \in S_1,v \in S_2\}|\\
        &= \binom{n}{2} \cdot (4n+7) + \ell \cdot |\{uv \in E \mid u \in S_1,v \in S_2\}|\text{.}
    \end{align*}
    Dividing by $\ell$ yields $|\{uv \in E \mid u \in S_1,v \in S_2\}| \ge k$ as $\ell > \binom{n}{2} \cdot (4n+7)$ and $k$ is an integer.
    
    \textbf{Runtime.}
    Note that by \Cref{prop:opt_clique} any full degree-$\Delta$ caterpillar $C$ for $K_{n+r}$ is optimal for $K_{n+r}$ by the symmetry of $K_{n+r}$. Hence, we can efficiently compute $\opt_{\Cater_\Delta}(K_{n+r})$ by fixing any such caterpillar $C$ and then simply computing $\cost(K_{n+r},C)$. The rest of the reduction runs in polynomial time.
\end{proof}
%}

%{\color{red} Leon's Vorschlag: definition of full caterpillars}
%A $\Delta$-legged \emph{caterpillar} $C$ is a connected graph
%consisting of a path $S$ - called the \emph{spine} of $C$ - and for 
%each spine vertex $s \in V(S)$, a fresh set of \emph{leg vertices} 
%$\leg(s)$ of size at most $\Delta$ whose elements are all connected 
%to $s$. In other words, 
%\begin{align*}
%    V(C) &:= V(S) \sqcup \bigsqcup_{s \in V(S)} \leg(s), \\
%    E(C) &:= 
%    E(S) \cup \bigcup_{s \in V(S)}  \{\{s,l\} ~ |~ l \in \leg(s)\}.
%\end{align*}
%Let $\Cater_\Delta$ be the class of all $\Delta$-legged caterpillars. 
%We call a $\Delta$-legged caterpillar $C$ \emph{full} if all but one 
%outermost spine vertex have $\Delta$ legs. Let $\FullCater_\Delta$ be 
%the class of all full, $\Delta$-legged caterpillars.
%
%\begin{claim}
%    For all $n \in \NN$, any optimal caterpillar $C$ in $K_{n}$, 
%    that is, 
%    \[
%        C \in \argmin \{\cost(K_{n}, C') \mid 
%        C' \in \Cater_\Delta, |V(C')| = n\} ,
%    \] is full and any full caterpillar is optimal in $K_n$.
%\end{claim}
%\begin{proof}
%    We prove this by strong induction on $n$. 
%    Clearly $C=K_1$ is a full caterpillar.
%    Suppose $C$ is an optimal $\Delta$-legged caterpillar 
%    for $K_{n}$, $n >1$. Assume that we can delete one edge of the 
%    spine of $C$ so that the size of one of the resulting 
%    caterpillars is divisible by $\Delta + 1$. 
%\end{proof}

\section{Designing Caterpillars for Trees}\label{sec:tic}

Recall that our general approximation approach provides a 4-approximation for \gicShort{} if the input graph is a tree (see \Cref{corollary:factor-4-approx}); we will subsequently refer to this special case of \gic{} as \tic{} (\ticShort{}).
An immediate question is, whether turning to approximation to solve \ticShort{} is necessary?
% An immediate question is, whether we can do better? 
% As our generic approximation is not tailored to trees, the answer is probably yes. 
While the polynomial-time algorithms for MLA for trees are non-trivial, the structure of minimum linear arangements of trees is by now well-understood~\cite{chung1984optimal,goldberg1976algorithm,shiloach1979minimum}.
In preliminary considerations, we were able to transfer a few of these structural observations, but not all.
% Given the close connection between \gicShort{} and MLA which we already exploited in \cref{sec:approx}, it seems thus plausible to transfer also the polynomial-time algorithm for trees.
Somewhat surprisingly, we cannot transfer all observations as we prove \ticShort{} to be \NP-hard.

A natural question is why the techniques underlying the tractability of MLA on trees do not carry over to caterpillar hosts.
All three mentioned algorithms for MLA on trees follow a divide-and-conquer approach:
they decompose the input tree into subtrees, recursively compute MLAs for the subtrees, and concatenate the resulting arrangements into an MLA for the whole tree.
Crucially, this strategy relies on the property that an MLA for the whole tree, restricted to any subtree in the decomposition, is again an MLA for that subtree.
This property fails for caterpillar hosts.
While an optimum caterpillar for a path-like subtree may again be path-like in isolation, in the context of the whole tree it can be beneficial to \emph{compress} such a caterpillar---that is, to move vertices off the spine and attach them as legs---to shorten edges with one endpoint in the subtree and one endpoint outside.
Consequently, an optimum caterpillar for the whole tree need not restrict to an optimum caterpillar for a subtree, which obstructs a direct divide-and-conquer approach.

We reduce from the strongly $\NP$-hard problem \textsc{$3$-Partition} (see §4.2.2, Theorem 4.4 in the book by Garey and Johnson~\cite{garey1979computers}).
\OptproblemDef{3-Partition} 
{Positive integers $a_1,\dots,a_{3n}$ and $B \in \NN$ satisfying $B / 4 < a_i < B/ 2$.}
{Is there a partition of the index set $\{1, \dots ,3n\}$ into parts $P_1, \dots ,P_n$ such that $|P_j| = 3$ and $\sum_{i \in P_j} a_i = B$ for all $1 \leq j \leq n$?}
% \todo[inline]{LKu: Why talk about multisets and partitions of multisets, and mention that the size of the parts is 3 only implicitly? Personally, I think giving the $a_i$'s as a list and talking about a partition of the index set cleaner. No need to mention multisets. Also, the proofs use the latter convention throughout.\\AN: I am fine with a different notation. But the one used before was inconsistent: ``partition of the index set $\{1, \dots , 3n\}$ into parts $P_1, \dots ,P_n$ such that $P_j = \{a_{j_1},a_{j_2},a_{j_3}\}$'' What contains~$P_j$? Indices or some $a_i$'s? \\ 
% Alternative proposal: Partition of the index set $I = \{1, \dots , 3n\}$ into parts $I_1, \dots ,I_n$ of size three each, so that each~$P_j = \{a_\ell \mid \ell \in I_j\}$ contains numbers that add up to~$B$, i.e., $\sum_{a \in P_j} a = B$}
% Note that the condition~$B / 2 < a_i < B/ 4$ implies that any part~$P$ contains exactly three elements of the multiset---hence the name of the problem.
% \textsc{3-Partition} asks whether, given integers $a_1,\dots,a_{3n}$ and $B$ that satisfy $B / 2 < a_i < B/ 4$ for all $1 \leq i \leq 3n$, there is a partition of the index set $\{1, \dots , 3n\}$ into parts $P_1, \dots ,P_n$ such that $P_j = \{a_{j_1},a_{j_2},a_{j_3}\}$ and $a_{j_1} + a_{j_2} + a_{j_3} = B$ for all $1 \leq j \leq n$. 
The general idea of our reduction is straightforward. 
We construct a tree $T$ that contains $3n$ disjoint subtrees $A_1, \dots, A_{3n}$ of sizes $a_1,\dots, a_{3n}$, respectively, and $n$ subtrees $B_1, \dots, B_n$, each of size $\kappa_2+1$ (here~$\kappa_2$ is a number much larger than~$B$). 
We choose the degree bound to be $\Delta = \kappa_2 + B$. 
The goal now is to adapt $T$ and $W$ so that any caterpillar~$C$ with $\cost(T,C) \leq W$ will have $n$ segments containing all vertices from one subtree $B_j$ and the vertices from three subtrees $A_{j_1}, A_{j_2}, A_{j_3}$. 
Thus, any partition of the index set satisfying the \textsc{3-Partition} conditions will correspond to the segments of a caterpillar~$C$ that incurs cost~$\cost(T,C) = W$ and vice versa.
% if and only if $\cost(T,C) \leq W$. In fact, the costs of an embedding of $T$ into $C$ will be exactly $W$. % mention tight costs

Most of the proof deals with the construction of $T$. 
In order to force the roots (called $\beta_1,\dots, \beta_n$ in the proofs) of $B_1, \dots, B_n$ to all be spine vertices (and thus not appear in the same segment), we give them so many legs that if we would not choose them as spine vertices, the cost would exceed $W$. 
A simple swapping argument then gives that their legs must also be in the same segment, that is, each $B_j$ is a subtree of a segment. 
As we want the vertices of $A_1,\dots,A_{3n}$ to be legs of these segments as well, we need to exclude the possibility of them forming segments themselves. 
This is achieved by introducing high-degree subtrees (called $B_1^*,\dots ,B_{\kappa_1}^*$ in the proofs) to $T$ which each induces a segment of $C$ for the same reasons $B_1, \dots, B_n$ do. 
Then, forming a segment solely for some $A_i$, would make the spine too long and again the cost would exceed $W$. 
Hence, $C$ can be assumed to have the intended structure mentioned above together with additional segments for $B_1^*,\dots ,B_{\kappa_1}^*$. 
Consult \Cref{fig:construction_T} and \Cref{fig:Tree-into-Caterpillar_C} for visualizations of the tree $T$ and the (intended) caterpillar $C$, respectively. 
The bound $W$ is chosen as the cost of the intended caterpillar for $T$.
% The backward direction of the correctness proof (\Cref{lm:Tree-into-Caterpillar_<=}) heavily uses scaling of the \textsc{3-Partition} instance to assure certain properties of $T$.

\begin{theorem}
	\tic{} is \NP-hard.
%     Given a tree $T$ and $\Delta, W \in \NN$, it is $\NP$-hard to find a $\Delta$-legged (or degree-$\Delta$) caterpillar $C$ such that $\cost(T,C) \leq W$.
    \label{thm:Tree-into-Caterpillar}
%      \todo{LKu: We need to make the distinction between given degree vs. given number of legs explicit}
\end{theorem}

Instead of proving this theorem directly, we opt for the more convenient proof where $\Delta$ is an upper bound on the legs of the caterpillar instead of an upper bound on its degree, that is, we show \NP-hardness for the problem of---given a tree $T$ and $\Delta, W \in \NN$---finding a $\Delta$-legged caterpillar $C$ for $T$ with cost at most $W$. This avoids discussing the edge cases for the outermost segments which can have one more leg than the inner segments if the degree is bounded. Later, we remark how this proof implies \NP-hardness for \ticShort{}. 

To prove this, let us begin by constructing $T$ and choosing $\Delta, W \in \NN$. 
Suppose integers $a_1,\dots,a_{3n}$ and $B$ for some $n \in \NN$ such that $B/ 4 < a_i < B/2$ are given. 
We assume without loss of generality that $a_1 + \dots + a_{3n} = nB$; otherwise, a desired partition of the index set can not exist.

Define
\begin{align*}
    \sigma (n) := 
    \begin{cases}
        \frac{(n-1)^2}{4} - \frac{n-1}{2} \quad &\text{if } n \text{ is odd} \\
        \frac{(n-1)^2}{4} \quad &\text{if } n \text{ is even}
    \end{cases} \quad \text{and} \quad 
    \lambda (n) := 3(\sigma(n) + n-1)
\end{align*}
and let
\begin{align*}
    \kappa_0 &:= n B, \quad
    \kappa_1 := 2((\kappa_0 - 3n) + \lambda(n) + 1) + 1 \\
    \kappa_2 &:= (\kappa_0 - 3n) + \sigma(n + \kappa_1) + \lambda(n) + B + 1\text{.}
\end{align*}
We will come back to the semantics of these numbers later.

We will first construct the tree $T$. For all $1 \leq i \leq 3n$, let $A_i$ be a new tree of depth 1 with root $\alpha_i$ and $a_i -1$ leaves. For all $1 \leq j \leq n$, let $B_j$ be a new tree of depth 1 with root $\beta_j$ and $\kappa_2$ leaves.
For all $1 \leq k \leq \kappa_1$, let $B^*_k$ be a new tree of depth 1 with root $\beta^*_k$ and $\kappa_2 + B - 3n -1$ leaves.
We will call the leaves of $A_i$, $B_j$ and $B^*_k$ \emph{$\alpha_i$-leaves}, \emph{$\beta_j$-leaves} and \emph{$\beta^*_k$-leaves}, respectively. 
We now obtain $T$ by connecting $\beta_1$ with $\alpha_1,\dots,\alpha_{3n}$, $\beta_2,\dots,\beta_n$ and $\beta^*_1,\dots,\beta_{\kappa_1}^*$, that is, 
\begin{align*}     
    V(T):=& \bigcup_{i = 1}^{3n} V(A_i) \cup \bigcup_{j = 1}^{n} V(B_j) \cup \bigcup_{k = 1}^{\kappa_1} V(B^*_k), \\
    E(T):=& \bigcup_{i = 1}^{3n} E(A_i) \cup \bigcup_{j = 1}^{n} E(B_j) \cup \bigcup_{k = 1}^{\kappa_1} E(B^*_k) \\
    &\cup \{\beta_1\alpha_i ~|~ 1 \leq i \leq 3n\} 
    \cup \{\beta_1\beta_j ~|~ 2 \leq j \leq n\} 
    \cup \{\beta_1\beta^*_k ~|~ 1 \leq k \leq \kappa_1\}.
\end{align*}
\begin{figure}
    \centering
    \begin{tikzpicture}
    %% T
    % Node layer 2
    \node[bvertex] (b1) at (5.25,2.1) {};
    
    % Node layer 1
    \node[avertex] (a1) at (0,1) {};
    \node (am) at (1,1) {$...$};
    \node[avertex] (a3n) at (2,1) {};

    \node[bvertex] (b11) at (3.1,1) {};
    \node (b1m) at (3.5,1) {$...$};
    \node[bvertex] (b12) at (3.9,1) {};

    \node[bvertex] (b2) at (5,1) {};
    \node (bm) at (6,1) {$...$};
    \node[bvertex] (bn) at (7,1) {};

    \node[bsvertex] (bs1) at (8.5,1) {};
    \node (bsm) at (9.5,1) {$...$};
    \node[bsvertex] (bsn) at (10.5,1) {};
    
    % Node layer 0 (leaves)
    \node[avertex] (a11) at (-0.4,0) {};
    \node (a1m) at (0,0) {$...$};
    \node[avertex] (a12) at (0.4,0) {};

    \node[avertex] (a3n1) at (1.6,0) {};
    \node (a3nm) at (2,0) {$...$};
    \node[avertex] (a3n2) at (2.4,0) {};

    \node[bvertex] (b21) at (4.6,0) {};
    \node (b2m) at (5,0) {$...$};
    \node[bvertex] (b22) at (5.4,0) {};

    \node[bvertex] (bn1) at (6.6,0) {};
    \node (bnm) at (7,0) {$...$};
    \node[bvertex] (bn2) at (7.4,0) {};

    \node[bsvertex] (bs11) at (8.1,0) {};
    \node (bs1m) at (8.5,0) {$...$};
    \node[bsvertex] (bs12) at (8.9,0) {};

    \node[bsvertex] (bsn1) at (10.1,0) {};
    \node (bsnm) at (10.5,0) {$...$};
    \node[bsvertex] (bsn2) at (10.9,0) {};

    % edge layer 2 -- 1 
    \foreach[count = \i] \a in {a1, am, a3n, b11, b1m, b12, b2, bm, bn, bs1, bsm, bsn}
    {
		\draw (b1) to[out=205 + \i * 10,in=90,looseness=0.2] (\a);
    }
%     \draw (b1) -- (a1);
%     \draw (b1) -- (am);
%     \draw (b1) -- (a3n);
% 
%     \draw (b1) -- (b11);
%     \draw (b1) -- (b1m);
%     \draw (b1) -- (b12);
% 
%     \draw (b1) -- (b2);
%     \draw (b1) -- (bm);
%     \draw (b1) -- (bn);
% 
%     \draw (b1) -- (bs1);
%     \draw (b1) -- (bsm);
%     \draw (b1) -- (bsn);
    
    % edge layer 1 -- 0
    \draw (a1) -- (a11);
    \draw (a1) -- (a1m);
    \draw (a1) -- (a12);

    \draw (a3n) -- (a3n1);
    \draw (a3n) -- (a3nm);
    \draw (a3n) -- (a3n2);

    \draw (b2) -- (b21);
    \draw (b2) -- (b2m);
    \draw (b2) -- (b22);

    \draw (bn) -- (bn1);
    \draw (bn) -- (bnm);
    \draw (bn) -- (bn2);

    \draw (bs1) -- (bs11);
    \draw (bs1) -- (bs1m);
    \draw (bs1) -- (bs12);

    \draw (bsn) -- (bsn1);
    \draw (bsn) -- (bsnm);
    \draw (bsn) -- (bsn2);

    % names
    \node (T) at (0,2.2) {$T:$};
    \node (b1n) at (4.8,2.2) {$\beta_1$};
    \draw [
        decoration={
            brace,
            mirror,
            raise=0.3cm
        },
        decorate
    ] (3,1) -- (4,1) node [pos=0.5,anchor=north,yshift=-0.4cm] {\shortstack{$\Delta - B$\\many}};

    \node (a1n) at (-0.3,1) {$\alpha_1$};
    \draw [
        decoration={
            brace,
            mirror,
            raise=0.3cm
        },
        decorate
    ] (-0.5,0) -- (0.5,0) 
    node [pos=0.5,anchor=north,yshift=-0.4cm] {\shortstack{$a_1 -1$ \\ many}};
    
    \node (a3nn) at (2.4,1) {$\alpha_{3n}$};
    \draw [
        decoration={
            brace,
            mirror,
            raise=0.3cm
        },
        decorate
    ] (1.5,0) -- (2.5,0) 
    node [pos=0.5,anchor=north,yshift=-0.4cm] {\shortstack{$a_{3n} -1$\\ many}};
    
    \node (b2n) at (4.7,1) {$\beta_2$};
    % \draw [
    %     decoration={
    %         brace,
    %         mirror,
    %         raise=0.3cm
    %     },
    %     decorate
    % ] (4.5,0) -- (5.5,0) 
    % node [pos=0.5,anchor=north,yshift=-0.4cm] {\shortstack{$\Delta - B$ \\ many}};
    
    \node (bnn) at (7.3,1) {$\beta_n$};
    \draw [
        decoration={
            brace,
            mirror,
            raise=0.3cm
        },
        decorate
    ] (6.5,0) -- (7.5,0) 
    node [pos=0.5,anchor=north,yshift=-0.4cm] {\shortstack{$\Delta - B$ \\ many}};

    \node (bs1n) at (8.1,1) {$\beta^*_1$};
    % \draw [
    %     decoration={
    %         brace,
    %         mirror,
    %         raise=0.3cm
    %     },
    %     decorate
    % ] (8,0) -- (9,0) 
    % node [pos=0.5,anchor=north,yshift=-0.4cm] {\shortstack{$\Delta-3n -1$ \\many}};

    \node (bsnn) at (10.9,1) {$\beta^*_{\kappa_1}$};
    \draw [
        decoration={
            brace,
            mirror,
            raise=0.3cm
        },
        decorate
    ] (10,0) -- (11,0) 
    node [pos=0.5,anchor=north,yshift=-0.4cm] {\shortstack{$\Delta -3n -1$ \\many}};
    \end{tikzpicture}
    \caption{The tree~$T$ constructed in our reduction.}
    \label{fig:construction_T}
\end{figure}
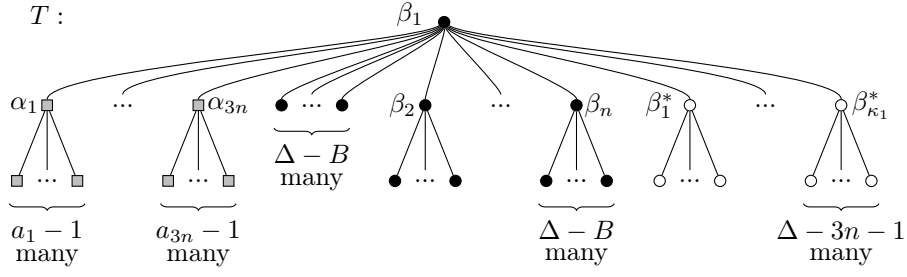
Consult \Cref{fig:construction_T} for a visualization of the tree $T$.
We define
\begin{align*}
    W:= |E(T)|  + \sigma(n+\kappa_1) + \lambda(n) + (\kappa_0 -3n) \quad \text{and} \quad 
    \Delta := \kappa_2 + B.
\end{align*}
The bound $W$ will turn out to be exactly the cost of the intended caterpillar $C$ for $T$ constructed in the proof of \Cref{lm:Tree-into-Caterpillar_=>} which is visualized in \Cref{fig:Tree-into-Caterpillar_C}.
    
We now prove the forward direction of the reduction, that is, for a given \textsc{3-Partition} YES-instance, we explicitly give a $\Delta$-legged caterpillar $C$ such that $\cost(T,C) = W$.
\begin{lemma}
    If there is a partition of $\{1, \dots, 3n\}$ into parts $P_1, \dots ,P_n$ such that $P_j = \{{j_1},{j_2},{j_3}\}$ and $a_{j_1} + a_{j_2} + a_{j_3} = B$ for all $1 \leq j \leq n$, then there is a $\Delta$-legged caterpillar $C$ with $\cost(T,C) = W$. 
    \label{lm:Tree-into-Caterpillar_=>}
\end{lemma}
\begin{proof}
    We construct a caterpillar $C$ with $V(C) = V(T)$ such that $\cost(T,C) = W$.
    Its spine $s_1\dots s_\ell$ will consist of the $\beta_j$ and $\beta^*_k$. 
    Thus, the spine will be of length $\ell = n + \kappa_1$. On the spine, $\beta_1$ is placed in the ``middle'', that is, $\beta_1 = s_{\lceil \frac{\ell}{2} \rceil}$. 
    Then, the $\beta_j$ are arranged as close as possible to $\beta_1$, that is,
    \[
        \{\beta_2,\dots,\beta_n\} 
        = \{s_k ~|~ \lceil \ell \slash 2 \rceil - \lfloor (n-1)\slash 2\rfloor \le k \le \lceil \ell\slash 2 \rceil + \lceil (n-1) \slash 2\rceil \} \setminus \{s_{\lceil \frac{\ell}{2} \rceil}\}.
    \]
%     \todo{LKu: If we need to cut more, than remove this math. The explanation + figure should be enough?}
    The $\beta^*_k$ can be placed arbitrarily as the rest of the spine vertices.
    For each $1 \leq k \leq \kappa_1$, the $\beta^*_k$-leaves are the legs of $\beta_k^*$. 
    For $1 \leq j \leq n$, the $\beta_j$-leaves are legs of $\beta_j$. Lastly, for every part $P_j = \{{j_1},{j_2}, {j_3}\},1 \leq j \leq n,$ of the partition of indices of the \textsc{$3$-Partition} YES-instance, place $\alpha_{j_1}, \alpha_{j_2}, \alpha_{j_3}$ together with their respective leaves as legs of some $\beta_{j'}$ (whose segment does not contain any $\alpha_i$ yet). 
    Now, each spine vertex has at most $\Delta$ many legs. Consult \Cref{fig:Tree-into-Caterpillar_C} for a visualization of $C$.
    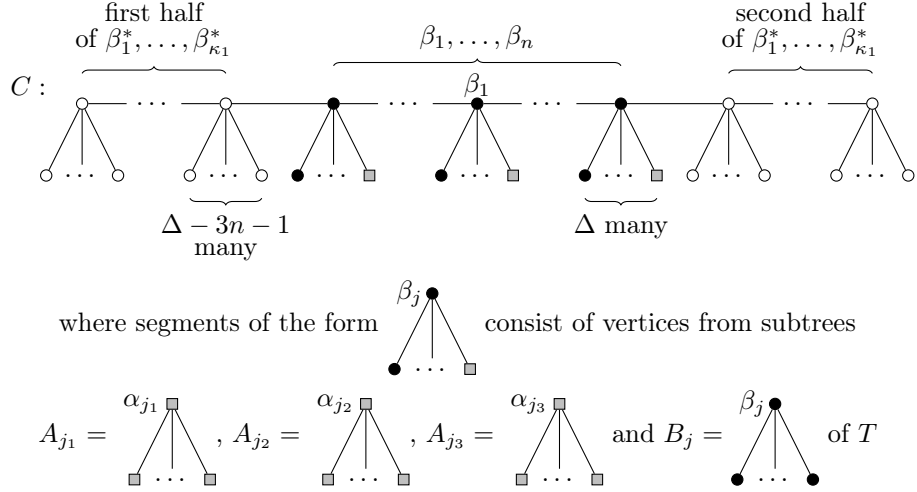
\begin{figure}
        \centering
        \begin{tikzpicture}[scale=0.95]
            % vertices
            %% spine 
            \node[bsvertex] (llbs) at (0.5,1) {};
            \node (llbs-rhelp) at (1.25,1) {};
        
            \node (lbs-dots) at (1.5,1) {$\dots$};
        
            \node (rlbs-lhelp) at (1.75,1) {};
            \node[bsvertex] (rlbs) at (2.5,1) {};
        
            \node[bvertex] (lb) at (4,1) {};
            \node (lb-rhelp) at (4.75,1) {};
        
            \node (lb-dots) at (5,1) {$\dots$};
        
            \node (mb-lhelp) at (5.25,1) {};
            \node[bvertex] (mb) at (6,1) {};
            \node (mb-rhelp) at (6.75,1) {};
        
            \node (rb-dots) at (7,1) {$\dots$};
            
            \node (rb-lhelp) at (7.25,1) {};
            \node[bvertex] (rb) at (8,1) {};
        
            \node[bsvertex] (lrbs) at (9.5,1) {};
            \node (lrbs-rhelp) at (10.25,1) {};
        
            \node (rbs-dots) at (10.5,1) {$\dots$};
        
            \node (rrbs-lhelp) at (10.75,1) {};
            \node[bsvertex] (rrbs) at (11.5,1) {};
        
            % legs 
            \node[bsvertex] (llbs-lleaf) at (0,0) {};
            \node (llbs-dots) at (0.5,0) {$\dots$};
            \node[bsvertex] (llbs-rleaf) at (1,0) {};
        
            \node[bsvertex] (rlbs-lleaf) at (2,0) {};
            \node (rlbs-dots) at (2.5,0) {$\dots$};
            \node[bsvertex] (rlbs-rleaf) at (3,0) {};
        
            \node[bvertex] (lb-lleaf) at (3.5,0) {};
            \node (lb-dots) at (4,0) {$\dots$};
            \node[avertex] (lb-rleaf) at (4.5,0) {};
        
            \node[bvertex] (mb-lleaf) at (5.5,0) {};
            \node (mb-dots) at (6,0) {$\dots$};
            \node[avertex] (mb-rleaf) at (6.5,0) {};
            
            \node[bvertex] (rb-lleaf) at (7.5,0) {};
            \node (rb-dots) at (8,0) {$\dots$};
            \node[avertex] (rb-rleaf) at (8.5,0) {};
        
            \node[bsvertex] (lrbs-lleaf) at (9,0) {};
            \node (lrbs-dots) at (9.5,0) {$\dots$};
            \node[bsvertex] (lrbs-rleaf) at (10,0) {};
        
            \node[bsvertex] (rrbs-lleaf) at (11,0) {};
            \node (rrbs-dots) at (11.5,0) {$\dots$};
            \node[bsvertex] (rrbs-rleaf) at (12,0) {};
        
            % edges
            %% spine
            \draw (llbs) -- (llbs-rhelp);
        
            \draw (rlbs) -- (rlbs-lhelp);
            \draw (rlbs) -- (lb);
        
            \draw (lb) -- (lb-rhelp);
            
            \draw (mb) -- (mb-lhelp);
            \draw (mb) -- (mb-rhelp);
        
            \draw (rb) -- (rb-lhelp);
            \draw (rb) -- (lrbs);
        
            \draw (lrbs) -- (lrbs-rhelp);
        
            \draw (rrbs) -- (rrbs-lhelp);
            
            %% legs
            \draw (llbs) -- (llbs-lleaf);
            \draw (llbs) -- (llbs-dots);
            \draw (llbs) -- (llbs-rleaf);
        
            \draw (rlbs) -- (rlbs-lleaf);
            \draw (rlbs) -- (rlbs-dots);
            \draw (rlbs) -- (rlbs-rleaf);
        
            \draw (lb) -- (lb-lleaf);
            \draw (lb) -- (lb-dots);
            \draw (lb) -- (lb-rleaf);
        
            \draw (mb) -- (mb-lleaf);
            \draw (mb) -- (mb-dots);
            \draw (mb) -- (mb-rleaf);
        
            \draw (rb) -- (rb-lleaf);
            \draw (rb) -- (rb-dots);
            \draw (rb) -- (rb-rleaf);
        
            \draw (lrbs) -- (lrbs-lleaf);
            \draw (lrbs) -- (lrbs-dots);
            \draw (lrbs) -- (lrbs-rleaf);
        
            \draw (rrbs) -- (rrbs-lleaf);
            \draw (rrbs) -- (rrbs-dots);
            \draw (rrbs) -- (rrbs-rleaf);
        
            % labels
            %% spine
            \node (bm-label) at (6,1.25) {$\beta_1$};
        
            \node (c) at (-0.25,1.25) {$C:$};
        
            \draw [
                decoration={brace,raise=0.4cm},decorate
            ] (0.5,1) -- (2.5,1) 
            node [pos=0.5,anchor=south,yshift=0.5cm] {
            \shortstack{first half \\of $\beta_1^*,\dots, \beta_{\kappa_1}^*$}
            };
        
            \draw [
                decoration={brace,raise=0.5cm},decorate
            ] (4,1) -- (8,1) 
            node [pos=0.5,anchor=south,yshift=0.6cm] {$\beta_1,\dots, \beta_{n}$};
        
            \draw [
                decoration={brace,raise=0.4cm},decorate
            ] (9.5,1) -- (11.5,1) 
            node [pos=0.5,anchor=south,yshift=0.5cm] {
            \shortstack{second half \\of $\beta_1^*,\dots, \beta_{\kappa_1}^*$}
            };
        
            %% legs
            % \draw [
            %     thick,decoration={brace,mirror,raise=0.2cm},decorate
            % ] (0,0) -- (1,0) 
            % node [pos=0.5,anchor=north,yshift=-0.3cm] {\shortstack{$\Delta - 3n -1$ \\ many}};
        
            \draw [
                decoration={brace,mirror,raise=0.3cm},decorate
            ] (2,0) -- (3,0) 
            node [pos=0.5,anchor=north,yshift=-0.4cm] {\shortstack{$\Delta - 3n -1$ \\ many}};
        
            % \draw [
            %     thick,decoration={brace,mirror,raise=0.2cm},decorate
            % ] (3.5,0) -- (4.5,0) 
            % node [pos=0.5,anchor=north,yshift=-0.3cm] {$\Delta$ many};
        
            % \draw [
            %     thick,decoration={brace,mirror,raise=0.2cm},decorate
            % ] (5.5,0) -- (6.5,0) 
            % node [pos=0.5,anchor=north,yshift=-0.3cm] {$\Delta$ many};
        
            \draw [
                decoration={brace,mirror,raise=0.3cm},decorate
            ] (7.5,0) -- (8.5,0) 
            node [pos=0.5,anchor=north,yshift=-0.4cm] {$\Delta$ many};
        
            % \draw [
            %     thick,decoration={brace,mirror,raise=0.2cm},decorate
            % ] (9,0) -- (10,0) 
            % node [pos=0.5,anchor=north,yshift=-0.3cm] {\shortstack{$\Delta - 3n -1$ \\ many}};
        
            % \draw [
            %     thick,decoration={brace,mirror,raise=0.2cm},decorate
            % ] (11,0) -- (12,0) 
            % node [pos=0.5,anchor=north,yshift=-0.3cm] {\shortstack{$\Delta - 3n -1$ \\ many}};
        \end{tikzpicture} \\
        where segments of the form
        \begin{tikzpicture}[baseline=0.5cm]
            \node[bvertex] (root) at (0.5,1) {};
            \node[bvertex] (lleaf) at (0,0) {};
            \node (dots) at (0.5,0) {$\dots$};
            \node[avertex] (rleaf) at (1,0) {};
        
            \draw (root) -- (lleaf);
            \draw (root) -- (dots);
            \draw (root) -- (rleaf);
        
            \node (label) at (0.2,1) {$\beta_j$};
        \end{tikzpicture}
        consist of vertices from subtrees \\
        $A_{j_1} = 
        \begin{tikzpicture}[baseline=0.5cm]
            \node[avertex] (root) at (0.5,1) {};
            \node[avertex] (lleaf) at (0,0) {};
            \node (dots) at (0.5,0) {$\dots$};
            \node[avertex] (rleaf) at (1,0) {};
        
            \draw (root) -- (lleaf);
            \draw (root) -- (dots);
            \draw (root) -- (rleaf);
        
            \node (label) at (0.1,1) {$\alpha_{j_1}$};
        \end{tikzpicture}$,
        $A_{j_2}=
        \begin{tikzpicture}[baseline=0.5cm]
            \node[avertex] (root) at (0.5,1) {};
            \node[avertex] (lleaf) at (0,0) {};
            \node (dots) at (0.5,0) {$\dots$};
            \node[avertex] (rleaf) at (1,0) {};
        
            \draw (root) -- (lleaf);
            \draw (root) -- (dots);
            \draw (root) -- (rleaf);
        
            \node (label) at (0.1,1) {$\alpha_{j_2}$};
        \end{tikzpicture}$,
        $A_{j_3}=
        \begin{tikzpicture}[baseline=0.5cm]
            \node[avertex] (root) at (0.5,1) {};
            \node[avertex] (lleaf) at (0,0) {};
            \node (dots) at (0.5,0) {$\dots$};
            \node[avertex] (rleaf) at (1,0) {};
        
            \draw (root) -- (lleaf);
            \draw (root) -- (dots);
            \draw (root) -- (rleaf);
        
            \node (label) at (0.1,1) {$\alpha_{j_3}$};
        \end{tikzpicture}$ 
        and $B_j = 
        \begin{tikzpicture}[baseline=0.5cm]
            \node[bvertex] (root) at (0.5,1) {};
            \node[bvertex] (lleaf) at (0,0) {};
            \node (dots) at (0.5,0) {$\dots$};
            \node[bvertex] (rleaf) at (1,0) {};
        
            \draw (root) -- (lleaf);
            \draw (root) -- (dots);
            \draw (root) -- (rleaf);
        
            \node (label) at (0.2,1) {$\beta_j$};
        \end{tikzpicture}$ 
        of $T$
        \caption{The intended caterpillar $C$ for~$T$.}
        \label{fig:Tree-into-Caterpillar_C}
    \end{figure}
    It remains to show that $\cost(T,C)=W$.
    To see this, first observe that the cost $|E(T)|$ is inevitable for any caterpillar $C$.
    As the edges between some $\alpha_i$ to any of its leaves get stretched to paths of length 2, they induce additional cost 1. There are $\sum_{i=1}^{3n} (a_i -1) = \kappa_0 - 3n$ such edges. 
    The cost $\sigma(n + \kappa_1)$ arise from the stretching of edges of the form $\beta_1\beta_j$ and $\beta_1\beta^*_k$; they are Gauss sums over the distances from $\beta_1$ to all the $\beta_j,\beta^*_k$ on the left and on the right of $\beta_1$, respectively.
    Note that they do not include the cost already covered by the summand $|E(T)|$ and that they are sensitive to the position of $\beta_1$---and thus also to the parity of $n + \kappa_1$. 
    It is now easy to check that these cost are precisely $\sigma(n + \kappa_1)$. 
    Lastly, consider the cost induced by edges of the form $\beta_1\alpha_i$. 
    As the $\alpha_i$ appear as triples as legs of some $\beta_j$ and the cost induced by the $\alpha_j$ that are $\beta_1$-legs are covered by $|E(T)|$, we get the additional cost of $\lambda(n)$. 
    The edges between $\beta_j, \beta^*_k$ and their leaves do not get stretched. Thus, $\cost (T,C) = W$.
\end{proof}

Before proving the converse direction, we give two simple lemmas about optimal caterpillars for graphs.
The first lemma says that in an optimal caterpillar,
% (of any graph into any other)
the spine vertex of a segment has maximal degree among all the vertices in that segment.
\begin{lemma}[Spine-Degree-Lemma]%; \appref{lm:degree_lemma}]
    Let $G$ be a graph and let $C$ be an optimal caterpillar for $G$. Let $S$ be a segment of $C$ with spine node $s \in V(S)$. Then, the following holds:
    \begin{itemize}
        \item $\deg_G (s) = \max \{ \deg_G (v) ~|~ v \in V(S)\}$, and
        \item swapping $s$ with $v \in V(S)$ for which $\deg_G(s) = \deg_G (v)$ holds also yields an optimal caterpillar for $G$.
    \end{itemize}
    \label{lm:degree_lemma}
\end{lemma}
% \appendixproof{lm:degree_lemma}{
\begin{proof}
    Swapping $s$ with a node $u \in V(S)$ for which $\deg_G(s) > \deg_G (u)$ holds would increase the cost by $\deg_G(s)$ but only decrease them by $\deg_G (u)$ if $su \notin E(G)$; if $su \in E(G)$, then the cost increase by $\deg_G(s)-1$ and decrease by $\deg_G(u)-1$. Hence, $\deg_G (s) = \max \{ \deg_G (v) ~|~ v \in V(S)\}$. Swapping $s$ with $v \in V(S)$ for which $\deg_G(s) = \deg_G (v)$ holds leaves the cost unchanged. Thus, the resulting caterpillar of such a swap is also optimal.
\end{proof}
% }

The second lemma characterizes minimal linear arrangements of star graphs up to isomorphism.
\begin{lemma}[Embedding Stars into Paths]%; \appref{lm:star_into_path}]
  \label{lm:star_into_path}
    Let $S = (V,E)$ be a star, that is, $V = \{c\}\cup \{v_1,\dots,v_{n-1}\}$ and $E = \{cv_i ~|~ 1\leq i \leq n-1\}$.
    Then, the optimal embeddings of $S$ into a path $p_1 \dots p_n$ are the embeddings where 
    \[ c = 
    \begin{cases}
        p_{\lceil n/2 \rceil} \quad &\text{if } n \text{ is odd}, \\
        p_{n/2} \text{ or } p_{n/2+1} &\text{if } n \text{ is even}.
    \end{cases}
    \]
\end{lemma}
% \appendixproof{lm:star_into_path}{
\begin{proof}
    Let $S$ be optimally embedded into the path $p_1 \dots p_n$. Note that by symmetry, the cost and thus the optimality of an embedding is determined by the position of $c$. 
    Assume $c$ is not at a position as specified in the statement of the lemma. Without loss of generality assume $c= p_j$ with $j < n/2$.
    Then, swapping $c$ with the vertex $p_{j+1}$ stretches the edges from $c$ to $p_1,\dots, p_{j-1}$ and thus produces additional cost $j-1$ but shortens the distance to $p_{j+2},\dots, p_n$, decreasing the cost by $n-(j+1)$. 
    As $j < n/2$,
    \[j-1 < n-(j+1) \quad \Longleftrightarrow \quad 2j <n\]
    and thus such a swap improves the embedding  which is a contradiction to the assumption. 
    As swapping improves the embedding until $c$ is in the ``middle'' of the path, that is, $c$ is at the position specified in the statement of the lemma, $c$ must be at one of these positions. 
\end{proof}
% }

Now, we are ready to prove the converse direction. We show that if there is a $\Delta$-legged caterpillar $C$ with $\cost(T,C) \leq W$, then we can assume $C$ to have the same structure as in the forward direction.
This allows us to determine a partition of the index set from legs of specific segments. Moreover, the cost must be tight, that is, $\cost(T,C) = W$. 
\begin{lemma}
    If $T$ is constructed and $\Delta, W \in \NN$ are chosen from a given \textsc{3-Partition} instance $\{a_1,\dots,a_{3n}\}$ and $B$ as above and there is a caterpillar $C$ such that $\cost(T,C) \leq W$, then there is a partition of $\{1, \dots, 3n\}$ into parts $P_1, \dots ,P_n$ such that $P_j = \{{j_1},{j_2},{j_3}\}$ and $a_{j_1} + a_{j_2} + a_{j_3} = B$ for all $1 \leq j \leq n$. Moreover, $\cost(T,C) = W$.
    \label{lm:Tree-into-Caterpillar_<=}
\end{lemma}
\begin{proof}
  Let $(T, W, \Delta)$ be an instance of \ticShort{} as constructed in \cref{lm:Tree-into-Caterpillar_=>}.
  Assume there is a $\Delta$-legged caterpillar $C$ with $V(C) = V(T)$ and $\cost(T,C) \leq W$.
    Our goal is to show that the corresponding \textsc{3-Partition} instance $a_1,\dots,a_{3n}$ and $B$ is a YES-instance. We will first sketch the proof and then prove the individual steps. Let $\mathcal{B} = \{\beta_1,\dots,\beta_n, \beta^*_1,\dots,\beta^*_{\kappa_1}\}$.
    \begin{enumerate}[label=(\alph*)]
        \item The set of spine vertices of $C$ contains $\mathcal{B}$. This will be a consequence of the fact that no two $\beta, \beta' \in \mathcal{B}$ are in the same segment of $C$ together with the Spine-Degree-Lemma (\Cref{lm:degree_lemma}).
        \item Now that we know that $\mathcal{B}$ is a subset of the spine vertices of $C$, we may assume that for all $\beta \in \mathcal{B}$, the $\beta$-leaves must be in the segment of $C$ that has $\beta$ as spine vertex.
        \item Using \Cref{lm:star_into_path}, we show that no $\alpha_i$ is a spine vertex.
        \item By (b) and the choice of $\Delta$, having an $\alpha_i$ as a leg of some $\beta^*_k$ would exceed the cost $W$. Hence, this can not happen.
        \item By scaling, we show that at most three $\alpha_i$ can be in the same segment of $C$. By (d), we know that exactly three $\alpha_i$ are in the same segment of $C$ whose spine vertex is some $\beta_j$.
        \item Finally, as $\cost (T,C) \leq W$, we can prove that all $\alpha_i$-leaves are in the segment that contains $\alpha_i$---otherwise, we would have $\cost (T,C) > W$. As each spine vertex has at most $\Delta = \kappa_2 + B$ many legs, $\kappa_2$ of them are occupied by the $\beta_j$-leaves and all $\alpha_i$ and their leaves are in the same segment, each segment $S_k$ of $C$ with $\{\alpha_{i_1},\alpha_{i_2},\alpha_{i_3}\} \subseteq S_k$ corresponds to the sum $a_{i_1} + a_{i_2} + a_{i_3} = B$. Therefore, we can conclude that $a_1,\dots,a_{3n}$ and $B$ is a YES-instance of \textsc{3-Partition}.
    \end{enumerate}
    Note that the last step also implies that $\beta_1$ is in the ``middle'' of the spine and the $\beta_j$ are then arranged as close as possible to $\beta_1$. 
    So, any caterpillar $C'$ such that $\cost(T,C') \leq W$ can be assumed to be of the form as in the first part of the correctness proof (\Cref{lm:Tree-into-Caterpillar_=>}).
    
    We now show the individual steps. In the following, by scaling $a_1,\dots,a_{3n}$ and $B$, we assume without loss of generality that
    (1) $B \ge 3n+1$,
    (2) $a_i \ge \lambda(n) + 3n + 4$,
    (3) $a_i \ge B/4 + \lambda(n) + 1$
%     \todo{LKu: Readable like that or is (enumerated) align better? Numbering is important and referenced later.}
    where we get the third assumption by scaling with $4(\lambda(n) + 1)$, as initially $a_i \ge B/4 + 1/4$.
    
    % step 1
    \begin{claim}[a]
        There are no two $\beta , \beta' \in \mathcal{B}$ that are mapped to the same segment $S$.
        \label{lm:no_2_beta(*)_same_segment}
    \end{claim}
    \begin{proof}
        Assume for contradiction that there are $\beta , \beta' \in \mathcal{B}$ with $\beta , \beta' \in V(S)$. The subtrees of $\beta , \beta'$ together have at least $2\kappa_2+2$ many vertices and at most $|V(S)|=\kappa_2+B+1$ many of them are in $V(S)$. 
        By assumption, $\beta , \beta' \in V(S)$. Thus, at least $\kappa_2 - B + 1$ many leaves of $\beta , \beta'$ are in other segments. 
        Each of these incurs cost of at least two, that is, additional cost of at least one. 
        We have $\cost(T,C) \ge |E(T)| + \kappa_2-B+1 > W$ by choice of $\kappa_2$. Thus, $\cost(T,C) > W$ which is a contradiction.
    \end{proof}
    
    By \Cref{lm:no_2_beta(*)_same_segment} and the Spine-Degree-Lemma (\Cref{lm:degree_lemma}), we may assume that all $\beta \in \mathcal{B}$ are spine vertices of $C$. 
    Next, we show that we may assume that all leaves of each $\beta \in \mathcal{B}$ are in the segment containing $\beta$.
    
    % step 2
    \begin{claim}[b]%,\appref{lm:beta(*)_leaves_together}]
        We may assume that for each $\beta \in \mathcal{B}$ all $\beta$-leaves are in the segment $S$ that contains~$\beta$.
        \label{lm:beta(*)_leaves_together}
    \end{claim}
%    \appendixproof{lm:beta(*)_leaves_together}
%    {
    \begin{proof}
        Let $\beta \in \mathcal{B}$ and $S$ be the segment that contains $\beta$. 
        Let $\hat\beta$ be a leaf of $\beta$ with $\hat\beta \notin V(S)$. 
        If $S$ does not have $\Delta$ legs yet, then making $\hat\beta$ a new leg of $S$ strictly decreases the cost. 
        Otherwise, there is at least one leaf $\ell \in V(T)$, that is not a $\beta$-leaf, that is a leg of $S$ as there are only $3n$ non-leaves in $T$, excluding $\beta' \in \mathcal{B}$ (namely $\alpha_i$) and as $B \ge 3n+1$ by (1). 
        By swapping $\ell$ with $\hat\beta$, the cost decreases by $d$ and increases by at most $d$, where $d$ denotes the distance of their spine vertices in $C$.
        By iteratively repeating this argument, the claim follows.
    \end{proof}
%    }
    Next, we show that no $\alpha_i$ is a spine vertex. That is, all $\alpha_i$ are legs of $C$.
    
    % step 3
    \begin{claim}[c]%,\appref{lm:no_alpha_spine}]
        There is no $\alpha_i$ that is a spine vertex.
        \label{lm:no_alpha_spine}
    \end{claim}
%    \appendixproof{lm:no_alpha_spine}
%    {
    \begin{proof}
        Assume towards a contradiction that $\alpha_i$ is a spine vertex. 
        By the Spine-Degree-Lemma (\Cref{lm:degree_lemma}), no $\beta \in \mathcal{B}$ is a leg of $\alpha_i$. 
        Thus, $C$ has at least $n + \kappa_1 + 1$ many segments containing neighbors of $\beta_1$ (or $\beta_1$ itself). 
        The farthest of them incurs an additional cost of $\left\lceil \frac{n + \kappa_1}{2} \right\rceil - 1$. 
        By \Cref{lm:star_into_path}, we have $\cost(T,C) \ge |E(T)| + \sigma(n + \kappa_1) + \left\lceil \frac{n + \kappa_1}{2} \right\rceil - 1 > W$ by choice of $\kappa_1$. Thus, $\cost(T,C) > W$ which is a contradiction.
    \end{proof}
%    }
    By the Spine-Degree-Lemma (\Cref{lm:degree_lemma}), it follows that no $\alpha_i$ is a leg of an $\alpha_i$-leaf that is a spine vertex.
    Next, we show that no $\alpha_i$ is a leg of $\beta_j^*$ and that at most three $\alpha_i$ are in a segment with spine vertex $\beta_{j'}$.
    
    % step 4
    \begin{claim}[d]%,\appref{lm:alpha_not_under_beta*}]
        No $\alpha_i$ is a leg of $\beta_j^*$.
        \label{lm:alpha_not_under_beta*}
    \end{claim}
%    \appendixproof{lm:alpha_not_under_beta*}
%    {
    \begin{proof}
        Assume towards a contradiction that there is an $\alpha_i$ that is a leg of $\beta_j^*$ and let $S$ be the segment that contains $\beta_j^*$. 
        By (2), the subtree of $\alpha_i$ has $a_i \ge \lambda(n) + 3n + 4$ many vertices and by \Cref{lm:beta(*)_leaves_together} at most $3n+1$ many of them are legs in $S$. 
        By assumption, $\alpha_i \in V(S)$. Thus, at least $\lambda(n) + 3$ many leaves of $\alpha_i$ are in other segments. 
        The cost incurred by the edges between $\alpha_i$ and its leaves is at least two and for at least $\lambda(n) + 1$ many at least three (as two $\alpha_i$-leaves could be the neighboring spine vertices), that is, additional cost of at least one and at least two, respectively. %The term $(\kappa_0 - 3n)$ counts both of these cost with weight one and the term $\lambda(n) + 1$ counts the rest of the additional cost.
        We have $\cost(T,C) \ge |E(T)| + \sigma(n + \kappa_1) + (\kappa_0 - 3n) + \lambda(n) + 1 > W$. Thus, $\cost(T,C) > W$ which is a contradiction.
    \end{proof}
%    }
    % step5
    \begin{claim}[e]%,\appref{lm:max_3_alphas_under_beta}]
        At most three $\alpha_i$ are in the same segment $S$ with spine vertex $\beta_j$.
        \label{lm:max_3_alphas_under_beta}
    \end{claim}
%	\appendixproof{lm:max_3_alphas_under_beta}
%	{
    \begin{proof}
        Assume towards a contradiction that there are $\alpha_{i_1},\alpha_{i_2},\alpha_{i_3},\alpha_{i_4} \in V(S)$ for a segment $S$ with spine vertex $\beta_j$. 
        By (3), the subtrees $A_{i_k}$ of $\alpha_{i_k}$, $1 \leq k \leq 4$ together have at least $B + 4\lambda(n) + 4$ many vertices and by \Cref{lm:beta(*)_leaves_together} at most $B$ many of them are legs in $S$. 
        By assumption, $\alpha_{i_k} \in V(S)$, $1 \leq k \leq 4$. Thus, at least $4\lambda(n) + 4$ many leaves of $\alpha_{i_k}$, $1 \leq k \leq 4$ are in other segments. 
        The cost incurred by the edges between $\alpha_{i_k}$, $1 \leq k \leq 4$ and their leaves is at least two and for at least $4\lambda(n) + 2$ many of them at least three, that is, additional cost of at least one and at least two, respectively.
        We have $\cost(T,C) \ge |E(T)| + \sigma(n + \kappa_1) + (\kappa_0 - 3n) + 4\lambda(n) + 2 > W$. Thus, $\cost(T,C) > W$ which is a contradiction.
    \end{proof}
%    }
    By \Cref{lm:no_alpha_spine}, \Cref{lm:alpha_not_under_beta*} and \Cref{lm:max_3_alphas_under_beta} and the fact that there are exactly $n$ many segments with spine vertex $\beta_j$, we conclude that exactly three $\alpha_i$ are in each segment with a spine vertex $\beta_j$.
    
    % step 6
    \begin{claim}[f]%,\appref{lm:alpha_leaves_in_alphas_segment}]
        For all $1 \leq i \leq 3n$, if $\alpha_i \in S$ for some segment $S$, then all $\alpha_i$-leaves are in~$S$.
        \label{lm:alpha_leaves_in_alphas_segment}
    \end{claim}
%	\appendixproof{lm:alpha_leaves_in_alphas_segment}
%	{
    \begin{proof}
        Let $c$ be the additional cost induced by edges from the $\alpha_i$ to their leaves. 
        As $c = (\kappa_0 -3n)$ only if each $\alpha_i$ and their leaves get mapped to the same segment, and we already have cost $\geq |E(T)| + \sigma(n + \kappa_1) + \lambda(n)$ by the previous claims and $\cost (T,C) \leq W$, the claim must be true.
    \end{proof}
%    }
    Let $S$ be a segment with spine node $\beta_j$ containing $\alpha_{j_1},\alpha_{j_2},\alpha_{j_3}$ and their leaves. As $S$ also contains all $\kappa_2$ many $\beta_j$-leaves and we have the upper bound $\Delta = \kappa_2 + B$ on the maximal number of legs each segment can have, $\alpha_{j_1},\alpha_{j_2},\alpha_{j_3}$ and their leaves take up at most $B$ many leg positions. As $a_1 + \dots + a_{3n} = nB$, we have $a_{j_1} + a_{j_2} + a_{j_3} = B$. Hence, we can read off the parts  $P_1, \dots, P_n$ of the partition from the legs of the segments $\seg(\beta_1),\dots, \seg(\beta_n)$, that is, 
    $ P_j := \{{j_1}, {j_2}, {j_3}\}$ if $\alpha_{j_1},\alpha_{j_2},\alpha_{j_3} \in \seg(\beta_j)$.
\end{proof}

To complete the proof of \Cref{thm:Tree-into-Caterpillar} (in the $\Delta$-legged case), observe that we can construct the instance $(T, W, \Delta)$ from the \textsc{3-Partition} instance $a_1,\dots,a_{3n}$ and $B$ in polynomial time as \textsc{3-Partition} is strongly \NP-hard. 

\begin{remark}
    To see why \Cref{thm:Tree-into-Caterpillar} also holds if we bound the degree instead of the maximal number of legs, we set the degree bound to $\Delta + 2$.
    Then, carefully scanning the proof above shows that they almost hold for the degree bounded case as well. 
    The only way the proof breaks is when the \textsc{$3$-Partition} instance contains $a_{i_1}, a_{i_2}, a_{i_3}$ such that $a_{i_1} + a_{i_2} + a_{i_3} = B+1$. 
    Then, we could give $\alpha_{i_1}, \alpha_{i_2}, \alpha_{i_3}$ and all but one of their leaves as legs to an outermost spine vertex $s$. The remaining leaf is then added as a new spine vertex next to $s$ on the spine. 
    This breaks \Cref{lm:alpha_leaves_in_alphas_segment} in the proof above. 
    However, we can easily circumvent this edge case by scaling the \textsc{$3$-Partition} instance by $2$. Then, $a_{i_1} + a_{i_2} + a_{i_3}$ is even and thus $a_{i_1} + a_{i_2} + a_{i_3} \neq B+1$.
\end{remark}

\begin{remark}
    Notably, the above proof also shows \NP-hardness for a variant of \tic{} where the caterpillar $C'$ is given in the input and the task is to find a bijection $\pi : V(T) \to V(C')$ that minimizes the cost function $\sum_{uv \in E(T)} \dist_C (\pi(u), \pi(v))$. 
    Now, choosing $T$ as above and $C'$ isomorphic to $C$ yields the desired reduction from \textsc{3-Partition}.
\end{remark}

\section{Conclusion}\label{sec:conclusion}

We studied \gicShort{}, a generalization of minimum linear arrangement (MLA), where the goal is to design a caterpillar $H$ with $V(H)=V(G)$ and maximum degree $\Delta$ (instead of a path)
that minimizes \(\sum_{uv\in E(G)}\dist_H(u,v)\).
On the one side, we show that we can lift any $\alpha$-approximation for MLA to an (\(\alpha+3-2/(\Delta-1)\))-approximation for \gicShort.
On the other side the problem remains \NP-hard
when $\Delta \ge 2$ is constant,
and when $G$ is a tree.
The latter result depicts a stark contrast to MLA, which is polynomial-time solvable when $G$ is a tree~\cite{chung1984optimal}.
We conclude with some future research directions.

First, can the approximation ratios be improved (or complemented by lower bounds)?
For general graphs, any approximation better than \(O(\sqrt{\log n}\log\log n)\) would rely on an entirely new approach. 
For trees, the situation seems more hopeful:
MLA is polynomial-time solvable here~\cite{chung1984optimal}, so the constant \(4\) comes entirely from the additive \(\Theta(|E(G)|)\) slack in~\starfy{}, which vanishes on instances where the optimum is \(\omega(|E(G)|)\).
A tailored algorithm exploiting the tree structure of \(G\) directly might yield a smaller constant.

The second research direction revolves around the degree bound $\Delta$.
Is \gicShort{} on trees fixed-parameter tractable with respect to $\Delta$?
And what happens, if we drop the degree bound?
Then we would always be ``allowed'' to choose a star as our caterpillar; however this need not be optimal.
The complexity of this case is left open, even if $G$ is a tree.

Looking at the bigger picture, one could study the problem of designing caterpillars for graphs under other objectives, such as the initially mentioned bandwidth or cutwidth measures.
Finally, considering other classes of graphs (as has been done in previous work~\cite{raspaud2000congestion,bezrukov2000congestion,avin2020demand,schmid2015splaynet,kellerhals2026designingapproximatebinarytrees})
may be a fruitful research avenue.

% \bibliography{references}

\bibliographystyle{splncs04}
\bibliography{references}

% \appendix

% \newpage

% \section{Missing Proof Details}

% \appendixProofs

% \appendixProofText

\end{document}